\documentclass[10pt, conference, compsocconf]{IEEEtran}
\usepackage{cite}
\usepackage[cmex10]{amsmath}
\usepackage{array}
\usepackage{fixltx2e}
\usepackage{url}

\usepackage{amssymb}
\usepackage{latexsym}
\usepackage{graphicx}
\usepackage{xspace}
\usepackage[usenames]{xcolor}
\usepackage{wasysym}
\usepackage[final,
            bookmarks,
             bookmarksopen,
             colorlinks, 
             final,
             linkcolor=blue,
             citecolor=brown,
             pdfstartview=FitH]%
{hyperref}
\usepackage{paralist}
\usepackage{wrapfig}

\begin{document}
%
\title{Information Flow Analysis for a Dynamically Typed Functional Language
with Staged Metaprogramming}


\author{\IEEEauthorblockN{Martin Lester\\
Luke Ong}
\IEEEauthorblockA{
Department of Computer Science\\
University of Oxford\\
Oxford, UK\\
\{martin.lester, luke.ong\}@cs.ox.ac.uk}
\and
\IEEEauthorblockN{Max Schaefer}
\IEEEauthorblockA{School of Computer Engineering\\
Nanyang Technological University\\
Singapore\\
schaefer@ntu.edu.sg}
}


%


\maketitle

\begin{abstract}
Web applications written in JavaScript are regularly used for dealing
with sensitive or personal data.  Consequently, reasoning about their
security properties has become an important problem, which is made
very difficult by the highly dynamic nature of the language,
particularly its support for runtime code generation.  As a first step
towards dealing with this, we propose to investigate security analyses
for languages with more principled forms of dynamic code generation.
To this end, we present a static information flow analysis for a
dynamically typed functional language with prototype-based inheritance
and staged metaprogramming.  We prove its soundness, implement it and
test it on various examples designed to show its relevance to
proving security properties, such as noninterference, in JavaScript.
To our knowledge, this is the first fully static information flow analysis for a
language with staged metaprogramming, and the first formal
soundness proof of a CFA-based information flow analysis for a
functional programming language.

\end{abstract}

\begin{IEEEkeywords}
noninterference; staged metaprogramming; CFA;
information flow; dynamically typed languages; JavaScript;
static analysis
\end{IEEEkeywords}

%
\IEEEpeerreviewmaketitle


\definecolor{darkred}{rgb}{0.75,0,0}
\newcommand{\TODO}[1]{\textcolor{darkred}{\textbf{$\blacktriangleright$#1$\blacktriangleleft$}}}

\newif\ifdraft\draftfalse

\definecolor{DarkGreen}{RGB}{0,100,0}
\definecolor{DarkBlue}{RGB}{0,0,200}
\ifdraft
\newcommand{\mlchanged}[1]{{\color{DarkGreen} {#1}}}
\newcommand{\lochanged}[1]{{\color{red} {#1}}} 
\newcommand{\mschanged}[1]{{\color{Cyan} {#1}}}
\newcommand{\ml}[1]{{\color{DarkGreen}{[{#1}---Martin]}}}
\newcommand{\lo}[1]{{\color{red}{[{#1}---Luke]}}}
\newcommand{\ms}[1]{{\color{DarkBlue}{[{#1}---Max]}}}
\else
\newcommand{\mlchanged}[1]{{#1}}
\newcommand{\lochanged}[1]{{#1}}
\newcommand{\mschanged}[1]{{#1}}
\newcommand{\ml}[1]{}
\newcommand{\lo}[1]{}
\newcommand{\ms}[1]{}
\fi
   
   

\newcommand{\NT}[1]{\textit{#1}}
\newcommand{\Ctor}[1]{\texttt{#1}}
\newcommand{\AbsVal}{\NT{AbsVal}} 
\newcommand{\AbsVar}{\NT{AbsVar}} 
\newcommand{\Label}{\NT{Label}}   
\newcommand{\Name}{\NT{Name}}     
\newcommand{\FV}{\textrm{FV}}     

\newcommand{\pow}{\mathcal{P}}
\newcommand{\kw}[1]{\ensuremath{\mathop{\mbox{\small\rm\textsf{\textbf{#1}}}}}}
\newcommand{\hole}{\ensuremath{\mathop{\bullet}}}
\newcommand{\finmap}{\stackrel{\mbox{\scriptsize fin}}{\longrightarrow}}
\newcommand{\dom}{\ensuremath{\mathop{dom}}}

\newcommand{\red}[1]{\stackrel{#1}{\dashrightarrow}}
\newcommand{\ev}[1]{\xrightarrow{#1}}
\newcommand{\rlname}[1]{{\footnotesize(\textsc{#1})}}

\newcommand{\uline}[1]{\rule[0pt]{#1}{0.4pt}}
\newcommand\proto{\mathtt{\quot\hbox{\uline{2mm}}proto\hbox{\uline{2mm}}\quot}}
\newcommand{\Marker}{\NT{Marker}}
\newcommand{\lbl}{\mathit{lbl}}  
\newcommand\Ctxt{{\cal C}}
\newcommand\lowsec{\hbox{\scriptsize L}}
\newcommand\highsec{\hbox{\scriptsize H}}
\newcommand\mm{\mathfrak{m}}


\newcommand{\blank}{\underline{\hspace{0.5em}}}
\newcommand{\evstar}[1]{\xrightarrow{#1}\!\!^*}
\newcommand{\evfin}[2]{\xrightarrow{#1}\!\!^{#2}}
\newcommand{\unmark}{\mathop{\textit{unmark}}}   
\newcommand{\erase}[2]{\lfloor {#1} \rfloor_{#2}}

\newcommand{\direct}{\rightsquigarrow}
\newcommand{\indirect}{\looparrowright}
\newcommand{\nedirect}{\mbox{\rotatebox{45}{$\direct$}}}
\newcommand{\sedirect}{\mbox{\rotatebox{-45}{$\direct$}}}
\newcommand{\flowsto}{\leadsto}
\newcommand{\ifmodels}{\models_{\textsc{IF}}}
\newcommand{\lowin}{\stackrel{{\rm low}}{\longrightarrow}}
\newcommand{\highin}{\stackrel{{\rm high}}{\longrightarrow}}

\newcommand{\slamjs}{SLamJS\xspace}
\newcommand{\lambdajs}{$\lambda_\mathrm{JS}$}
\newcommand{\quot}{\texttt{"}}

\newcommand{\anystage}{{\scriptsize \wasylozenge}}
\newcommand{\ang}[1]{\langle {#1} \rangle}

\newcommand{\lsec}{\hbox{\scriptsize L}}
\newcommand{\hsec}{\hbox{\scriptsize H}}
\newcommand{\isec}{\hbox{\scriptsize I}}
\newcommand{\lowset}{\{ \lsec \}}
\newcommand{\vn}[1]{\mathit{#1}} 

\newcommand{\ee}{\epsilon}

\newtheorem{theorem}{Theorem}
\newtheorem{example}{Example}
\newtheorem{lemma}{Lemma}
\newtheorem{examplelist}[example]{\hspace*{-1\parindent}Example}


\section{Introduction}
\label{sec:intro}
An \emph{information flow} analysis determines which values in a
program can influence which parts of the result of the program.  Using
an information flow analysis, we can, for instance, prove that program
inputs that are deemed high security do not influence low security
outputs; this important security property is known as
\emph{noninterference}~\cite{goguen82}.

Early work on noninterference focused mainly on applications in a
military or government setting, where there might be strict rules
about security clearance and classification of documents.  More
recently, there has been increased interest in information security
(and hence its analysis) for Web applications, particularly for Web 2.0
applications written in JavaScript.

We have developed a static information flow analysis for a dynamically
typed, pure, functional language with stage-based metaprogramming~\cite{DBLP:conf/popl/KimYC06};
we call the language \slamjs (Staged Lambda JS) because it exhibits a
number of JavaScript's interesting features in an idealised, lambda
calculus-based setting~\cite{guha10}.  The analysis is based on the idea of
extending a constraint-based formulation of the analysis 0CFA~\cite{shivers88}
with constraints to track information
flow. We believe that the idea could be extended to other CFA-style
analyses (such as CFA2~\cite{vardoulakis11}) for improved precision.
We have formally proved the correctness of our analysis; we have also
implemented it and tested it on a number of examples.

Supporting material,
which includes mechanisations of our key results in the theorem prover Coq
and an implementation of our analysis in OCaml,
is available online at
\url{http://mjolnir.cs.ox.ac.uk/web/slamjs/}.

The structure of the remainder of the paper is as follows.  In
Section~\ref{sec:slamjs}, we present \slamjs: we begin with an
explanation of why we believe our chosen combination of language
features is relevant to information security in Web applications.
Next, we present the semantics of \slamjs and explain, using an
augmented semantics, what information flow means in this language.
Section~\ref{sec:analysis} explains how the analysis works and how we
proved its correctness.  We discuss our implementation and some
examples on which we have tested the analysis in
Section~\ref{sec:implementation}.  In Section~\ref{sec:related}, we
examine the gap between our work and a practical analysis for
real-world Web applications.  We also discuss other research
on analysis of information flow and staged metaprogramming,
before concluding in Section~\ref{sec:conclusions}.

\section{The Language \slamjs}\label{sec:slamjs}
\subsection{Motivation}
The new arena of Web applications presents many interesting challenges for
information flow analysis.
While there is an extensive body of research on information flow in
statically typed languages~\cite{DBLP:journals/toplas/PottierS03},
there is little tackling dynamically typed
languages.
The semantics of JavaScript are complex and poorly understood~\cite{DBLP:conf/aplas/MaffeisMT08},
which makes any formal analysis difficult.
Web applications frequently comprise code from multiple sources (including
libraries and adverts), written by multiple authors in an ad-hoc style. They
are often interactive (so cannot be viewed as a single execution with inputs
and outputs) and it might not be known in advance which code will be loaded.

The $\kw{eval}$ construct of JavaScript, which allows execution of arbitrary code
strings, is particularly troublesome, to the extent that many analyses
just ignore it.
However, a recent survey shows that real JavaScript code uses $\kw{eval}$
extensively~\cite{richards11}.
Its uses vary widely from straightforward (loading data via JSON)
through ill-informed (accessing fields of an object without using array notation)
to subtle (changing scoping behaviour)
and complex (emulating higher order functions).
We think that it is important to develop techniques for analysing this
notorious construct.

So that we might reasonably work formally, we have developed a
simplified language called \slamjs.  The language is heavily
influenced by \lambdajs, a ``core calculus'' for
JavaScript~\cite{guha10}.  Like JavaScript, \slamjs is dynamically
typed and features first-class functions and objects with
prototype-based inheritance.  Like JavaScript, it allows code to be
constructed, passed around and executed at run-time. Unlike
JavaScript, this is achieved using Lisp-style code quotations rather
than code strings~\cite{choi11}.  Recent work indicates that real-world usage of
$\kw{eval}$ is often of a form that could be expressed using code
quotations~\cite{jensen12}.  Thus analysis of programs with executable
code quotations is an important step towards analysis of programs with
executable code strings.

\subsection{Syntax and Semantics of \slamjs}

\subsubsection{Syntax}

\begin{figure*}
\[
\begin{array}{llcl}
\text{Booleans} & b & ::= & \kw{true} \mid \kw{false} \\
\text{Strings} & s & \in & \NT{String} \\
\text{Numbers} & n & \in & \NT{Number} \\
\text{Names} & x & \in & \Name \\
\text{Constants} & k & ::= & \kw{undef} \mid \kw{null} \mid b \mid s \mid n \\
\text{Expressions} & e & ::= & k \mid \{\overline{s:e}\} \mid x \mid \kw{fun}(x)\{e\} \mid e(e) \mid \kw{box}\>e \mid \kw{unbox}\>e  \mid \kw{run}\>e \\
             &   & \mid & \kw{if}(e)\{e\}\kw{else}\{e\} \mid e[e] \mid e[e]=e \mid \kw{del}\>e[e] \mid (e,\rho) \mid \kw{run}\>e\>\kw{in}\>\rho \\
\text{Values} & v,v^0 & ::= & (\kw{fun}(x)\{e\},\rho) \\
              & v^n & ::= & k \mid \{\overline{s:v^n}\} \mid (\kw{box}\>v^{n+1}) \\
              & v^{n+1} & ::= & x \mid (\kw{fun}(x)\{v^{n+1}\}) \mid (v^{n+1}(v^{n+1})) \mid (\kw{run}\>v^{n+1}) \\
              &        & \mid & (\kw{if}(v^{n+1})\{v^{n+1}\}\kw{else}\{v^{n+1}\}) \\    	      &        & \mid & (v^{n+1}[v^{n+1}]) \mid (v^{n+1}[v^{n+1}]=v^{n+1}) \mid  (\kw{del}\>v^{n+1}[v^{n+1}]) \\
              & v^{n+2} & ::= & (\kw{unbox}\>v^{n+1}) \\
\text{Environments} & \rho & \in & \Name\finmap v^0
\end{array}
\]
\caption{Syntax of \slamjs}
\label{fig:syntax}
\end{figure*}

\slamjs is a functional language with atomic constants, records, branching, first-class
functions and staged metaprogramming; the syntax is given in Figure~\ref{fig:syntax}.

The language has five types of atomic constant: booleans, strings,
numbers and two special values ($\kw{undef}$ and $\kw{null}$) to
indicate undefined or null values.  A record $\{\overline{s:v}\}$ is a
finite mapping from fields (named by strings) to values.  Fields can
be read ($e[e]$), updated or replaced ($e[e]=e$) and deleted
($\kw{del}\>e[e]$).  Records support prototype-based lookup: a read
from an undefined field of a record is redirected to the corresponding
field on the record held in its $\proto$ field, if there is one.

Branching on boolean values is enabled by the
$\kw{if}(e)\{e\}\kw{else}\{e\}$ construct.
Functions can be defined ($\kw{fun}(x)\{e\}$) and applied ($e(e)$).

Staged metaprogramming is supported through use of the $\kw{box}$,
$\kw{unbox}$ and $\kw{run}$ constructs. $\kw{box}\>e_1$ turns $e_1$ into a
``quoted'' or ``boxed'' code value, which can be executed using $\kw{run}$.
The use of $\kw{unbox}\>e_2$ within a boxed expression $e_1$
forces evaluation of $e_2$ to a boxed value, which is spliced into $e_1$
\emph{before it becomes a boxed value}.

Expressions of the form $(e,\rho)$ and $\kw{run}\>e\>\kw{in}\>\rho$
only arise as intermediate terms during execution: the former
represents an explicit substitution where all free variables of the
expression $e$ are given their value by the environment $\rho$; the
latter represents an expression to be unboxed and evaluated in
environment $\rho$.

Values exist at all stages. Constants, records with constant fields and
constant code quotations are values at every stage; closures are only
values at stage zero. Other constructs may be values at higher stages,
provided that their subexpressions are values at the appropriate
stage. We generally omit the stage superscript for values of
stage zero.

\subsubsection{Semantics}
We give a small-step operational semantics with evaluation contexts
and explicit substitutions for \slamjs.  There are two reduction
relations, $\red{n}$ and $\ev{n}$, each annotated with a level $n$.
The former is for top-level reduction, while the latter is for evaluation under
a context.

\emph{Evaluation contexts} In a staged setting, evaluation
contexts may straddle stage boundaries, hence they are annotated with
stage subscripts and superscripts.  A context $C^m_n$ denotes a hole
at stage $n$ inside an expression at stage $m$.  For a context $C^m_n$ and
an expression $e$, we denote by $C^m_n\langle e\rangle$ the expression
obtained by plugging $e$ into the hole contained in $C^m_n$.
The grammar of some key evaluation contexts is given in Figure~\ref{fig:somecontexts};
full details are
in Appendix~\ref{a:sem}.

\begin{figure*}
\[
\begin{array}{rclclcclcl}
C^m_n & ::=   & [\,] & \in & \Ctxt^n_n 
& \hspace{1em} \\
      & \mid & (\kw{fun}(x)\{C^{m+1}_n\}) & \in & \Ctxt^{m+1}_n 
&      & \mid & (\kw{if}(C^m_n)\{e\}\kw{else}\{e\}) & \in & \Ctxt^m_n \\
      & \mid & (C^m_n(e)) & \in & \Ctxt^m_n 
&      & \mid & (\kw{if}(v^{m+1})\{C^{m+1}_n\}\kw{else}\{e\}) & \in & \Ctxt^{m+1}_n \\
      & \mid & (v^m(C^m_n)) & \in & \Ctxt^m_n 
&      & \mid & (\kw{if}(v^{m+1})\{v^{m+1}\}\kw{else}\{C^{m+1}_n\}) & \in & \Ctxt^{m+1}_n \\
      & \mid & (\kw{unbox}\>C^m_n) & \in & \Ctxt^{m+1}_n 
&      & \mid & (\kw{box}\>C^{m+1}_n) & \in & \Ctxt^m_n \\
      & \mid & (\kw{run}\>C^m_n\>\kw{in}\>\rho) & \in & \Ctxt^m_n
&      & \mid & (\kw{run}\>C^m_n) & \in & \Ctxt^m_n
\end{array}
\]
\caption{Selected evaluation contexts}
\label{fig:somecontexts}
\end{figure*}

\emph{Reduction rules} Top-level reduction rules fall into two
categories: environment propagation rules for pushing explicit
substitutions inwards (Figure~\ref{fig:pushenv}),
and proper reduction rules (Figure~\ref{fig:proper reduction rules}). The
former are fairly straightforward, so full details are left for Appendix~\ref{a:sem}.
Note that explicit substitutions
only apply at stage zero, hence $(x,\rho)$ evaluates to $x$ at level
$n+1$ without looking up $x$ in $\rho$.
Furthermore, observe that $(\kw{run}\>e, \rho)$ pushes its environment into $e$,
allowing boxed code values to capture variables from outside.

\begin{figure*}
\[
\begin{array}{rclcrcl}
(k,\rho) & \red{n} & k & &
(x,\rho) & \red{n+1} & x \\
(\kw{fun}(x)\{e\},\rho) & \red{n+1} & (\kw{fun}(x)\{(e,\rho)\}) & &
(e_1 (e_2),\rho) & \red{n} & ((e_1,\rho)((e_2,\rho))) \\
(\kw{box}\>e,\rho) & \red{n} & (\kw{box}\>(e,\rho)) & &
(\kw{unbox}\>e,\rho) & \red{n} & (\kw{unbox}\>(e,\rho)) \\
(\kw{run}\>e,\rho) & \red{0} & (\kw{run}\>(e,\rho)\>\kw{in}\>\rho) & &
(\kw{run}\>e,\rho) & \red{n+1} & (\kw{run}\>(e,\rho)) \\
\multicolumn{3}{r}{(\kw{if}(e_1)\{e_2\}\kw{else}\{e_3\},\rho)}
& \red{n} &
\multicolumn{3}{l}{(\kw{if}((e_1,\rho))\{(e_2,\rho)\}\kw{else}\{(e_3,\rho)\})}
\end{array}
\]
\caption{Selected environment propagation rules}
\label{fig:pushenv}
\end{figure*}

The proper reduction rules are
also quite standard~\cite{choi11}, except for the field access
rules, which are designed to mimic JavaScript semantics as far as possible.

In particular, every record is expected to have a $\proto$ field,
which holds either the value $\kw{null}$ or another record, giving
rise to a chain of prototype objects that ultimately ends in
$\kw{null}$. Reading a record field follows this chain by rule
(\textsc{Read}2), until the field is either found (\textsc{Read}1), or
the top of the chain is reached, where (\textsc{Read}3) yields
\kw{undef}. Note that the reduction $\red{0}$ can get stuck, for
example, when applying a non-function, or branching on a non-boolean.

\begin{figure*}
\[
\begin{array}{lrcll}
\rlname{Lookup} & (x,\rho) & \red{0} & \rho(x) & \\ 
\rlname{Apply} & ((\kw{fun}(x)\{e\},\rho)(v)) & \red{0} & (e,\rho[x\mapsto v]) \\
\rlname{Unbox} & (\kw{unbox}\>(\kw{box}\>v^1)) & \red{1} & (v^1) \\
\rlname{Run} & (\kw{run}\>(\kw{box}\>v^1)\>\kw{in}\>\rho) & \red{0} & (v^1, \rho) \\
\rlname{IfTrue} & (\kw{if}(\kw{true})\{e_1\}\kw{else}\{e_2\}) & \red{0} & e_1 \\
\rlname{IfFalse} & (\kw{if}(\kw{false})\{e_1\}\kw{else}\{e_2\}) & \red{0} & e_2 \\
\rlname{Read1} & (\{\overline{s:v},s_i:v_i,\overline{s:v}'\}[s_i]) & \red{0} & v_i \\
\rlname{Read2} & (\{\overline{s:v},\proto:\{\overline{s:v}'\},\overline{s:v}''\}[s_x]) & \red{0} & (\{\overline{s:v}'\}[s_x]) & \text{if $s_x\not\in\overline{s}\cup\overline{s}''$} \\
\rlname{Read3} &  (\{\overline{s:v},\proto:\kw{null},\overline{s:v}''\}[s_x]) & \red{0} & \kw{undef} & \text{if $s_x\not\in\overline{s}\cup\overline{s}''$} \\
\rlname{Write1} & (\{\overline{s:v},s_i:v_i,\overline{s:v}'\}[s_i]=v_i') & \red{0} & \{\overline{s:v},s_i:v_i',\overline{s:v}'\} \\
\rlname{Write2} & (\{\overline{s:v}\}[s_x]=v_x) & \red{0} & \{\overline{s:v},s_x:v_x\} & \text{if $s_x\not\in\overline{s}$}\\
\rlname{Del1} & (\kw{del}\>\{\overline{s:v},s_i:v_i,\overline{s:v}'\}[s_i]) & \red{0} & \{\overline{s:v},\overline{s:v}'\} \\
\rlname{Del2} & (\kw{del}\>\{\overline{s:v}\}[s_x]) & \red{0} & \{\overline{s:v}\} & \text{if $s_x\not\in\overline{s}$}\\
\end{array}
\]

\caption{Proper reduction rules}\label{fig:proper reduction rules}
\end{figure*}

\begin{figure*}
\[
\begin{array}{lrcll}
\rlname{Lift-App} & ((\mm : e), \rho)(v) & \red{0} &
	(\mm : (( e, \rho)(v))) \\
\rlname{Lift-If} & (\kw{if} (\mm : v) \{ e_1 \} \kw{else} \{ e_2 \} ) & \red{0} &
	(\mm : (\kw{if} (v) \{ e_1 \} \kw{else} \{ e_2 \} )) \\
\rlname{Lift-Unbox} & \kw{unbox}\>(\mm : v) & \red{1} & (\mm : (\kw{unbox}\>v)) \\
\rlname{Lift-RunIn} & \kw{run}\>(\mm : v) \kw{in}\>\rho & \red{0} &
	(\mm : (\kw{run}\>v \kw{in}\>\rho)) \\
\rlname{Lift-ReadSel} & (v_1 [\mm : v_2]) & \red{0} & (\mm : (v_1[v_2])) \\
\rlname{Lift-ReadRec} & ((\mm : v_1) [v_2]) & \red{0} & (\mm : (v_1[v_2])) \\
\rlname{Lift-WriteSel} & (v_1 [\mm : v_2] = v_3) & \red{0} & (\mm : (v_1[v_2] = v_3)) \\
\rlname{Lift-WriteRec} & ((\mm : v_1) [v_2] = v_3) & \red{0} & (\mm : (v_1[v_2] = v_3)) \\
\rlname{Lift-DelSel} & (\kw{del}\>v_1[\mm : v_2]) & \red{0} & (\mm : (\kw{del}\>v_1[v_2]))\\
\rlname{Lift-DelRec} & (\kw{del}\>(\mm : v_1)[v_2]) & \red{0} & (\mm : (\kw{del}\>v_1[v_2]))
\end{array}
\]
\caption{Semantic rules for lifts}
\label{fig:lifts}
\end{figure*}

There is only a single rule for $\ev{m}$:
\[
\begin{array}{rcll}
C^m_n\langle e\rangle & \ev{m} & C^m_n\langle e'\rangle & \text{if $e\red{n}e'$}
\end{array}
\]

We write $\ev{\anystage}$ for the union over all $m$ of $\ev{m}$, and
$\evstar{\anystage}$ for its reflexive, transitive closure.

\begin{example}
\label{ex:langdef_simple_if}
Here is an evaluation trace of a simple $\kw{if}$ statement.
We use $\ee$ to stand for the empty environment.
\[
\begin{array}{ll}
 & (\kw{if} (\kw{true}) \{ \kw{false} \} \kw{else} \{ 1 \}, \ee) \\
\ev{0} & \kw{if} ((\kw{true}, \ee)) \{ (\kw{false}, \ee) \} \kw{else} \{ (1, \ee) \} \\
\ev{0} & \kw{if} (\kw{true}) \{ (\kw{false}, \ee) \} \kw{else} \{ (1, \ee) \} \\
\ev{0} & (\kw{false}, \ee) \ev{0} \kw{false}
\end{array}
\]
\end{example}

\begin{example}
The staging constructs in \slamjs allow fragments of code to be treated as
values and spliced together or evaluated at run-time, as shown in this
evaluation trace.
\[
\begin{array}{ll}
\mbox{\hspace{0.5em}} & (\kw{run}\> (\kw{box}\> (\kw{if} (\kw{unbox}\> (\kw{box}\> (\kw{true}))) \{\kw{false}\} \kw{else} \{1\})), \ee) \\
\lefteqn{\ev{0}} & \kw{run}\> (\kw{box}\> (\kw{if} (\kw{unbox}\> (\kw{box}\> (\kw{true}))) \{\kw{false}\} \\
    & \hspace{5em} \kw{else} \{1\}), \ee) \kw{in} \ee \\
\lefteqn{\evstar{0}} & \kw{run}\> (\kw{box}\> (\kw{if} (\kw{unbox}\> (\kw{box}\> (\kw{true}))) \{(\kw{false}, \ee)\} \\
    & \hspace{5em}  \kw{else} \{(1, \ee)\})) \kw{in}\> \ee \\
\lefteqn{\ev{0}} & \kw{run}\> (\kw{box}\> (\kw{if} (\kw{true}) \{(\kw{false}, \ee)\} \kw{else} \{(1, \ee)\})) \kw{in} \ee \\
\lefteqn{\evstar{0}} & \kw{run}\> (\kw{box}\> (\kw{if} (\kw{true}) \{\kw{false}\} \kw{else} \{1\})) \kw{in} \ee \\
\lefteqn{\ev{0}} & (\kw{if} (\kw{true}) \{\kw{false}\} \kw{else} \{1\}, \ee) \\
\lefteqn{\ev{0}} & \kw{if} (\kw{true}, \ee) \{(\kw{false}, \ee)\} \kw{else} \{(1, \ee)\} \\
\lefteqn{\ev{0}} & \kw{if} (\kw{true}) \{(\kw{false}, \ee)\} \kw{else} \{(1, \ee)\}
\ev{0} (\kw{false}, \ee)
\ev{0} \kw{false}
\end{array}
\]
\end{example}

\newcommand{\envmapsto}{\!\mapsto\!}
\begin{example}
Our staging constructs allow variables to be captured by code values
originating outside their scope.
Here, the code value $\kw{box}\>y$ is outside the scope of $y$,
but captures it during evaluation.
\[
\begin{array}{ll}
\mbox{\hspace{0.5em}} & (((\kw{fun}(x)\{(\kw{fun}(y)\{\kw{run}\> x\})\}) (\kw{box}\> y))(\kw{true}), \ee) \\
\lefteqn{\evstar{0}} & (\kw{fun}(y)\{\kw{run}\> x\}, \ang{x \envmapsto \kw{box}\> y}) (\kw{true}) \\
\lefteqn{\ev{0}} & (\kw{run}\> x, \ang{y \envmapsto \kw{true}, x \envmapsto \kw{box}\> y}) \\
\lefteqn{\ev{0}} & \kw{run}\> (x, \ang{y \envmapsto \kw{true}, x \envmapsto \kw{box}\> y}) \kw{in}\> \ang{y \envmapsto \kw{true}, x \envmapsto \kw{box}\> y} \\
\lefteqn{\ev{0}} & \kw{run}\> (\kw{box}\> y) \kw{in}\> \ang{y \envmapsto \kw{true}, x \envmapsto \kw{box}\> y} \\
\lefteqn{\ev{0}} & (y, \ang{y \envmapsto \kw{true}, x \envmapsto \kw{box}\> y}) \\
\lefteqn{\ev{0}} & \kw{true}
\end{array}
\]
This useful feature is vital for modelling certain uses of $\kw{eval}$;
the above code corresponds to this JavaScript: \\
\texttt{((function (x) \{return function (y) \{ \\
\mbox{} \mbox{} \mbox{} \mbox{} return (eval(x));\}\})("y"))(true);} \\
However, the power comes at a price:
the usual alpha equivalence property of
$\lambda$-calculus does not hold in \slamjs,
which makes reasoning about programs harder.
\end{example}

\subsection{Augmented Semantics of \slamjs}

The result of a program can depend on its component values in essentially two
different ways.
Consider programs operating on two variables $l$ and $h$.
The program
$(\kw{if}(l) \{ h \} \kw{else} \{ 1 \})$
may evaluate to the value of $h$ (if $l$ is $\kw{true}$); we say that there is a
\emph{direct flow} from $h$ to the result.
Conversely, the program
$(\kw{if} (h) \{ \kw{true} \} \kw{else} \{ 1 \})$
cannot evaluate to $h$. However, the result of evaluation tells us whether
$h$ was $\kw{true}$ or $\kw{false}$ because $h$ influences the control flow of the
program; there is an \emph{indirect flow} from $h$ to the result of the
program.

In order to track the dependency of a result on its component subexpressions,
we augment the language with explicit \emph{dependency
markers}~\cite{DBLP:conf/icfp/PottierC00,DBLP:conf/popl/AbadiBHR99}.
We also introduce new rules for lifting markers into their parent
expressions to avoid losing information about dependencies.
The augmented semantics is not intended for use in the execution of programs;
rather, we use it for analysing and reasoning about dependencies in the
original language. We begin by adding markers to the syntax:
\[
\begin{array}{llcl}
\text{Markers} & \mm & \in & \Marker \\
\text{Expressions} & e & ::= & \ldots \mid (\mm : e) \\
\text{Values} & v^n & ::= & \ldots \mid (\mm : v^n)
\end{array}
\]
We extend contexts to allow evaluation within a marked expression:
\[
\begin{array}{rclrclcl}
C^m_n & ::=    & \ldots 
      & \mid & (\mm : C^m_n) & \in & \Ctxt^m_n
\end{array}
\]
We allow propagation of environments within marked expressions:
\[
\begin{array}{rcl}
(\mm : e,\rho) & \red{n} & (\mm : (e,\rho))
\end{array}
\]

In Figure~\ref{fig:lifts} we introduce lifts to maintain a record of indirect flows.
Note there is no need for a lift rule on the right of an
assignment (i.e., $v_1[v_2] = (\mm:v_3) \red{0} (\mm: v_1[v_2] = v_3)$),
since the flow from $v_3$ is direct.

\begin{example}
Recall Example~\ref{ex:langdef_simple_if}. Suppose we add markers to each of
the components of the $\kw{if}$. The evaluation trace now becomes:
\[
\begin{array}{ll}
& (\kw{if} ((\hsec : \kw{true})) \{ (\lsec : \kw{false}) \} \kw{else} \{ ( \isec : 1) \}, \ee) \\
\evstar{0} & \kw{if} ((\hsec : \kw{true})) \{ ((\lsec : \kw{false}), \ee) \} \kw{else} \{ (( \isec : 1), \ee) \} \\
\ev{0} & (\hsec : (\kw{if} (\kw{true}) \{ ((\lsec : \kw{false}), \ee) \} \kw{else} \{ (( \isec : 1), \ee) \} )) \\
\ev{0} & (\hsec : ((\lsec : \kw{false}), \ee)) \\
\evstar{0} & (\hsec : (\lsec : \kw{false}))
\end{array}
\]
Note how the markers $\hsec$ and $\lsec$ in the result indicate that it depends on
the marked values $(\hsec : \kw{true})$ and $(\lsec : \kw{false})$.
\end{example}

\begin{example}
\label{ex:marked_fun}
Here is an example of marked evaluation with functions:
\[
(((\kw{fun} (x) \{ \isec : (\kw{fun} (y) \{ x \}) \})(\hsec : 1))(\lsec : 2), \ee)
\evstar{0} (\isec : (\hsec : 1))
\]
Observe that the result depends on $\isec$ because the function $(\isec : (\kw{fun} (y) \{ x \}) \}))$
was used to compute it,
but not on $\lsec$, as $(\lsec : 2)$ is discarded by that function.
\end{example}

\subsubsection*{Simulation}

Consider a function $\unmark$, defined in the obvious way, which strips an
expression of all markers. Clearly
if $\unmark(e_1) = f_1 \ev{n} f_2$, then for some $e_2$ such that $e_1 \evstar{n} e_2$, we have $\unmark(e_2) = f_2$.

\section{Information Flow Analysis for \slamjs}\label{sec:analysis}
\subsection{Overview}

Before we can define an information flow analysis,
we need to define what information flow is.
Following Pottier and Conchon~\cite{DBLP:conf/icfp/PottierC00},
we use the idea that if information does not flow from a marked expression into a value resulting from evaluation,
then erasing that marked expression or replacing it with a dummy value should not affect the result of evaluation.
(We use only their proof technique; their type-based analysis is not applicable to our language.)
We begin in Section~\ref{section:erasure} by defining erasure and establishing some results about its behaviour.

Our information flow analysis is built on top of a 0CFA-style analysis
capable of handling our staging constructs. Two variants of such an
analysis are explained in Section~\ref{section:0cfa};
mechanised correctness proofs in Coq are available online.

In Section~\ref{section:if}, we present the information flow analysis itself.
A key idea in CFA is that control flow influences data flow and vice versa.
Information flow depends on control and data flow, but the reverse is not
true. Therefore it is possible to treat information flow analysis as an
addition to CFA, rather than a completely new combined analysis.
We have two versions of the CFA, each of which yields an
information flow analysis.
We sketch a correctness proof of the simpler analysis;
complete mechanised proofs of both are available online.

Finally, in Section~\ref{section:soundness}, we prove soundness of the information flow
analysis. We also discuss its relationship with noninterference.

\subsection{Erasure and Stability}
\label{section:erasure}

\subsubsection{Erasure and Prefixes} We extend the language with a ``hole'' that
behaves like an unbound variable:

\[
\begin{array}{llcl}
\text{Expressions} & e & ::= & \ldots \mid \blank \\
\text{Values} & v^n & ::= & \ldots \mid \blank
\end{array}
\]

\[
\begin{array}{rcl}
(\blank ,\rho) & \red{n} & \blank
\end{array}
\]

Now for $M \subseteq \Marker$, define the $M$-\emph{erasure}
of $e$, written $\erase{e}{M}$, to be:
$e$ with any subexpression $(\mm : e')$
where $\mm \notin M$ replaced by $\blank$.
A full definition is in Appendix~\ref{a:sem}.

\subsubsection{Prefixing and Monotonicity}

We say that $e_1$ is a \emph{prefix} of $e_2$ or write $e_1 \preccurlyeq e_2$ if
replacing some subexpressions of $e_2$ with $\blank$ gives $e_1$.

Evaluation is monotonic with respect to prefixing:
if $e_1 \preccurlyeq e_2$ and $e_1 \evstar{\anystage} f$, where $f$ contains no $\blank$,
then $e_2 \evstar{\anystage} f$.

\begin{lemma}[Step Stability]
If $e_1 \red{n} e_2$, then
either $\erase{e_1}{M} \red{n} \erase{e_2}{M}$
or the reduction rule applied to derive this is a lift (\textsc{Lift-*}) of a marker $\mm \notin M$.
\end{lemma}
\begin{proof}
By induction over the rules defining $\red{n}$.
\end{proof}


\begin{theorem}[Stability]
Consider an expression $e_1$ (which may use $\blank$) and a $\blank$-free expression $e_2$
such that $e_1 \evstar{\anystage} e_2$.
Then for every $M \subseteq \Marker$ such that $\erase{e_2}{M} = e_2$,
it follows that $\erase{e_1}{M} \evstar{\anystage} \erase{e_2}{M}$.
\end{theorem}
\begin{proof}
Consider any $e_2$ and $M$ with $\erase{e_2}{M} = e_2$.
Aim to prove, for any $e_1$ with $e_1 \evstar{\anystage} e_2$, that $\erase{e_1}{M} \evstar{\anystage} e_2$.
Argue by induction over the length $k$ of derivations of $e_1 \evstar{\anystage} e_2$.

\emph{Base case:} $k = 0$. So $e_1 = e_2$. We have $\erase{e_2}{M} = e_2$,
so trivially $\erase{e_1}{M} = e_2$.

\emph{Inductive step:} $k = k' + 1$.
Given $e_1 \ev{n} e \evfin{\anystage}{k'} e_2$, aim to prove $\erase{e_1}{M} \evstar{\anystage} e_2$.
Assume by the induction hypothesis that $\erase{e}{M} \evfin{\anystage}{k'} e_2$.
Let $e_1 = C^m_n \ang{f_1}$ and $e = C^m_n \ang{f}$ with $f_1 \red{n} f$.
Case split on if $f_1 \red{n} f$ is a lift of a marker $\mm  \notin M$.

If it is such a lift, then let $f = (\mm  : f')$.
Now $\erase{f}{M} = \blank$, so $\erase{f}{M} \preccurlyeq \erase{f_1}{M}$.
Thus $\erase{C^m_n \ang{f}}{M} \preccurlyeq \erase{C^m_n \ang{f_1}}{M}$;
that is, $\erase{e}{M} \preccurlyeq \erase{e_1}{M}$.
We already have (from the induction hypothesis) that $\erase{e}{M} \evfin{\anystage}{k'} e_2$.
Now, applying Monotonicity, we get $\erase{e_1}{M} \evstar{\anystage} e_2$.

Otherwise, apply the Step Stability Lemma to get $\erase{f_1}{M} \red{n} \erase{f}{M}$.
It follows that $\erase{C^m_n \ang{f_1}}{M} \ev{n} \erase{C^m_n \ang{f}}{M}$;
that is, $\erase{e_1}{M} \ev{n} \erase{e}{M}$.
Using the induction hypothesis gives $\erase{e_1}{M} \ev{n} \erase{e}{M} \evfin{\anystage}{k'} e_2$, as required.
\end{proof}

\begin{example}
Recall that in Example~\ref{ex:marked_fun}, the result depended on $\hsec$ and $\isec$, but not $\lsec$.
Applying $\erase{-}{\{ \hsec, \isec \}}$ and evaluating the initial expression gives:
\[
(((\kw{fun} (x) \{ \isec : (\kw{fun} (y) \{ x \}) \})(\hsec : 1))( \blank ), \ee)
\evstar{0} (\isec : (\hsec : 1))
\]
That is, the result of evaluation is unchanged.
\end{example}

\subsection{0CFA for \slamjs}
\label{section:0cfa}
We use a context-insensitive, flow-insensitive control flow analysis
(0CFA~\cite{shivers88}) to approximate statically the set of values
to which individual expressions in a program may evaluate at runtime.

As far as 0CFA is concerned, the only non-standard feature of \slamjs
are its staging constructs. Roughly speaking, $\kw{box}$ and
$\kw{unbox}$/$\kw{run}$ act like function abstraction and application,
except that they use a dynamic (instead of static) scoping
discipline.\footnote{This intuition is made more precise in Choi et
  al.'s work on static analysis of staged programs~\cite{choi11},
  where staging constructs are translated into function abstraction
  and application; we prefer to work directly on the staged language
  for simplicity.}

We present two variants of 0CFA for \slamjs: a simple, but somewhat
imprecise formulation that does not distinguish like-named variables
bound by different abstractions, and a more complicated one that does.

\paragraph{Simple Analysis}
Following Nielson, Nielson and Hankin~\cite{nnh}, we formalise our
analysis by means of an \emph{acceptability judgement} of the form
$\Gamma,\varrho\models e$, where $\Gamma$ is an \emph{abstract
cache} associating sets of abstract values with labelled program
points, and $\varrho$ is an \emph{abstract environment} mapping
local variables and record fields to sets of abstract
values. Intuitively, the purpose of this judgement is to ensure that
$\Gamma(\ell)$ soundly over-approximates all possible values to which
the expression at program point $\ell$ can evaluate, and $\varrho$
does the same for variables and record fields.

More precisely, we assume that all expressions in the program are
labelled with labels drawn from a set $\Label$. An abstract cache is a
mapping $\Label\to\pow(\AbsVal)$ associating a set of abstract values
with every program point; similarly, an abstract environment
$\varrho\colon\AbsVar\to\pow(\AbsVal)$ maps abstract variables to sets of
abstract values, where an abstract variable is either a simple name
$x$ (representing a function parameter), or a field name of the form
$\ell.p$, where $\ell$ is a label representing a record, and $p$ is
the name of a field of that record.

Our domain of abstract values is mostly standard, with, e.g., an
abstract value $\mathtt{NULL}$ to represent the concrete $\kw{null}$
value, an abstract $\mathtt{NUM}$ value representing any number, and
abstract values $\mathtt{FUN}(x,e)$, $\mathtt{BOX}(e)$ and
$\mathtt{REC}(\ell)$ representing, respectively, a function value, a
quoted piece of code, and a record allocated at program point
$\ell$. (A complete definition of all our abstract domains is given in
Appendix~\ref{app:cfa}). For an abstract environment $\varrho$ and a label $\ell$ we define $\mathtt{proto}(\ell)_{\varrho}$ to be the smallest set $P\subseteq\Label$ such that
$\ell\in P$ and for every $p\in P$ and $\mathtt{REC}(\ell')\in\varrho(p.\proto)$ also $\ell'\in P$.

The acceptability judgement is now defined using syntax-directed
rules, some of which are shown in Figure~\ref{fig:acceptability part}
(the remaining rules, which are standard, are given in
Appendix~\ref{app:cfa}).

\begin{figure*}
\[
\begin{array}{lll}
\Gamma,\varrho\models k^\ell & \text{if} & \lceil k\rceil\in\Gamma(\ell) \\
\Gamma,\varrho\models x^\ell & \text{if} & \varrho(x)\subseteq\Gamma(\ell) \\
\Gamma,\varrho\models(\kw{box}\>e)^\ell & \text{if} & \Gamma,\varrho\models e \\
                                       & \text{and} & \exists \nu\in\Gamma(\ell).\Gamma,\varrho\models\nu\approx\kw{box}\>e \\
\Gamma,\varrho\models(\kw{unbox}\>e)^{\ell} & \text{if} & \Gamma,\varrho\models e \\
                                                  & \text{and} & \forall\mathtt{BOX}(e')\in\Gamma(\lbl(e)).\Gamma(\lbl(e'))\subseteq\Gamma(\ell) \\
\Gamma,\varrho\models(\kw{if}(e_1)\{e_2\}\kw{else}\{e_3\})^{\ell} & \text{if} & \Gamma,\varrho\models e_1 \land \Gamma,\varrho\models e_2 \land \Gamma,\varrho\models e_3 \\
                                                                                        & \text{and} & \Gamma(\lbl(e_2))\subseteq\Gamma(\ell) \land \Gamma(\lbl(e_3))\subseteq\Gamma(\ell)
\end{array}
\]
\caption{Some rules for the 0CFA acceptability judgement}\label{fig:acceptability part}
\end{figure*}

We write $t^\ell$ to represent an expression of the syntactic form
$t$, labelled with $\ell$. Thus, $k^\ell$ means an expression
consisting of a literal $k$ labelled $\ell$, and the first rule simply
says that in order for $\Gamma$ and $\varrho$ to constitute an
acceptable analysis of $k^\ell$, $\Gamma(\ell)$ must contain the
abstract value $\lceil k\rceil$ representing $k$. Similarly, the
second rule requires $\Gamma$ and $\varrho$ to be consistent in the
abstract values they assign to variables and references to them.  The
rules for dealing with function abstractions and records are standard
and so are elided here for brevity.

\begin{figure*}
\[\begin{array}{lll}
\Gamma,\varrho\models\lceil k\rceil\approx k & & \text{for any literal $k$} \\
\Gamma,\varrho\models\mathtt{FUN}(x,e)\approx \kw{fun}(x)\{e\} \\
\Gamma,\varrho\models\mathtt{BOX}(e)\approx \kw{box}\>e \\
\Gamma,\varrho\models\mathtt{REC}(\ell')\approx \{\overline{s:t^\ell}\} & \text{if} & \forall i.\exists \nu_i\in\varrho(\ell'.s_i). \Gamma,\varrho\models\nu_i\approx t_i\\
\Gamma,\varrho\models\nu\approx t' & \text{if} & \Gamma,\varrho\models \nu\approx t \land t^\ell\ev{n}t'^{\ell} \\
\Gamma,\varrho\models\nu\approx(t,\rho) & \text{if} & \Gamma,\varrho\models\nu\approx t \land \Gamma,\varrho\models\rho \\
\end{array}\]
\caption{The approximation judgement $\Gamma,\varrho\models\nu\approx t$}
\label{fig:approx-judgement}
\end{figure*}

The rule for $\kw{box} e$ requires $\Gamma$ and $\varrho$ to be an
acceptable analysis of the single sub-expression $e$, and for
$\Gamma(\ell)$ to include an abstract value $\nu$ approximating
$\kw{box} e$, which is written as
$\Gamma,\varrho\models\nu\approx\kw{box} e$. This judgement holds if
$\nu=\mathtt{BOX}(e)$, but we must be slightly more flexible: during
evaluation, unboxing may splice new code fragments into $e$, changing
its syntactic shape to some new expression $e'$. In order for
the flow analysis to be effectively computable, we want the set of
abstract values to be finite, so we cannot expect every such
$\mathtt{BOX}(e')$ to be part of our abstract domain. Instead, we
close the approximation judgement under reduction, that is, if
$\Gamma,\varrho\models\nu\approx t$ and $t^\ell\ev{n}t'^{\ell'}$, then
also $\Gamma,\varrho\models\nu\approx t'$; the full definition of the
approximation judgement appears in Figure~\ref{fig:approx-judgement}.

The rule for $\kw{unbox} e$ is surprisingly simple: all that is
required is that for any abstract value $\mathtt{BOX}(e')$ that the
analysis thinks can flow into $\lbl(e)$ (i.e., the label of expression
$e$) every abstract value flowing into its body $e'$ also flows into
the unboxing expression. Note that this models the name capture
associated with dynamic scoping, since our abstract environment
$\varrho$ does not distinguish between different variables of the
same name. The rule for $\kw{run}$ is the same as for $\kw{unbox}$.

Finally, we show the rule for $\kw{if}$, which is standard: any
abstract value that either of the branches can evaluate to is also a
possible result of the entire $\kw{if}$ expression.

To show this acceptability judgement makes sense, we prove its
coherence with evaluation:

\begin{theorem}[CFA Coherence]
If $\Gamma,\varrho\models e$ and $e\ev{n}e'$, then
$\Gamma,\varrho\models e'$.
\end{theorem}

The proof of this theorem is fairly technical and is elided here. A
full formalisation in Coq is available online in our supporting material.

Owing to its syntax-directed nature, the definition of the acceptability
relation can quite easily be recast as constraint rules; by generating
and solving all constraints for a given program, an acceptable flow
analysis can be derived.

Note that, while there may be infinitely many abstract values of the form
$\mathtt{BOX}(e)$ and $\mathtt{FUN}(e)$ that are relevant to a particular program,
the closure of the approximation judgement under reduction means that
the analysis need only consider those corresponding to subexpressions $e$ of the
original program, not those that may arise during execution.
That is, the analysis need only solve a finite set of constraints over a finite
set of abstract values and a finite set of labels and abstract variables,
so it can be guaranteed to terminate.

\begin{example}
Recall again Example~\ref{ex:marked_fun}. Our implementation of the analysis
labels the expression as follows:
\[
(((\kw{fun} (x) \{ (\isec : (\kw{fun} (y) \{ x^0 \})^1)^2 \})^3(\hsec : 1^4)^5)^6(\lsec : 2^7)^8)^9
\]
By generating and solving constraints it gives the following solution for $\Gamma$:
\[
\begin{array}{cclcclcclccl}
\lefteqn{0} & \lefteqn{\mapsto} & \{\mathtt{NUM}\} &
\lefteqn{1} & \lefteqn{\mapsto} & \{\mathtt{FUN}(y,(x)^0)\} &
\lefteqn{2} & \lefteqn{\mapsto} & \{\mathtt{FUN}(y,(x)^0)\} \\
\lefteqn{3} & \lefteqn{\mapsto} & \multicolumn{4}{l}{\{\mathtt{FUN}(x,((\isec :
(\kw{fun}(y)\{(x)^0\})^1)^2\lefteqn{)\}}} &
\lefteqn{4} & \lefteqn{\mapsto} & \{\mathtt{NUM}\} \\
\lefteqn{5} & \lefteqn{\mapsto} & \{\mathtt{NUM}\} &
\lefteqn{6} & \lefteqn{\mapsto} & \{\mathtt{FUN}(y,(x)^0)\} &
\lefteqn{7} & \lefteqn{\mapsto} & \{\mathtt{NUM}\} \\
\lefteqn{8} & \lefteqn{\mapsto} & \{\mathtt{NUM}\} &
\lefteqn{9} & \lefteqn{\mapsto} & \{\mathtt{NUM}\}
\end{array}
\]
while $\varrho = \{ x \mapsto \{\mathtt{NUM}\}, y \mapsto \{\mathtt{NUM}\} \}$.
As expected, the result of evaluation (labelled $9$) is a number.
\end{example}

\paragraph{Improved Analysis}
The analysis presented so far is not very precise, since abstract
environments do not distinguish identically named parameters
of different functions. Ordinarily, this is not a problem, as one
can rename them apart, but this is not possible for \slamjs, which
does not enjoy alpha conversion.

To restore analysis precision in the absence of alpha conversion, we introduce an
\emph{abstract context} $\Xi$ that keeps track of name bindings. In a
single-staged language, such an abstract context would simply map a
name $x$ to the innermost enclosing function abstraction whose
parameter is $x$. In a multi-staged setting, we need to distinguish
between bindings at different stages, hence the abstract context
maintains one such mapping per stage.
Thus $\Xi$ is a stack of frames, one for each stage;
a frame maps each variable name to the label of its binding context.

For instance, the two uses of $x$ in the \slamjs expression 
$\kw{fun}(x)\>\{\kw{box}(\kw{fun}(x)\>\{\>(\kw{unbox} x)(x)\>\})\}$
are at different stages, and hence bound by
different abstractions: the first $x$ by the outer abstraction, the
second by the inner one.

The acceptability judgement for the improved analysis is now of the
form $\Gamma,\varrho,\Xi\models e$, and the derivation rules
include additional bookkeeping to adjust $\Xi$ when analysing
subexpressions at different stages. While conceptually simple, this
change somewhat complicates the formalism, so we do not present it in
detail here; a full formalisation is available in the supporting material.

\subsection{Information Flow for \slamjs}
\label{section:if}

Assume we have already analysed a program using 0CFA
and found environments $\Gamma, \varrho$ that over-approximate the values
flowing to each labelled expression. We use information about which
functions and boxed values may occur to assist in determining what direct and indirect
flows occur between labels of the expression.

By recursing over the structure of an expression,
we generate constraints on a relation $\flowsto$:
\[
\flowsto \; : \; (\Label \uplus \Name \uplus \Marker) \rightarrow (\Label \uplus \Name \uplus \Marker)
\]
As an expression,
the labels, variable names and markers occurring within an expression
and the abstract values in the results of 0CFA for an expression
are all finite, the process will terminate.

We express constraints between labels, variable names and markers as either
direct flows ($x \direct y \implies x \flowsto y$)
or indirect flows ($x \indirect y \implies x \flowsto y$). 
(The distinction between direct and indirect is for clarity of exposition;
there is no practical difference between them with regard to the resulting analysis.)

Note that if instead we interpret $x \direct y$ and $x \indirect y$
as (elements of) relations and define ${\flowsto} = {\direct \cup \indirect}$,
then $\flowsto$ satisfies the constraints.

We say that $\Gamma, \varrho, \flowsto \ifmodels e$ if
$\Gamma, \varrho \models e$ and the conditions in Figure~\ref{fig:ifrules} hold.
As $\Gamma, \varrho$ and $\flowsto$ are constant
throughout the definition, we abbreviate
$\Gamma, \varrho, \flowsto \ifmodels e$ to $\ifmodels e$
for clarity.

\begin{figure*}
{
\[
\begin{array}{| l | l | l | l |}
\hline
\mbox{\emph{Expression} $e$} & \mbox{\emph{Subexpressions}} & \mbox{\emph{Direct Flows}} & \mbox{\emph{Indirect Flows}} \\
\mbox{$\ifmodels e$ holds:} & \mbox{if:} & \mbox{and:} & \mbox{and:} \\
\hline
\hline
k^\ell & - & - & - \\[0.3ex]
x^\ell & - & x \direct \ell & - \\[0.3ex]
\kw{fun}(x) \{ t^{\ell_1} \}^{\ell_2} & \ifmodels t^{\ell_1} & - & - \\[0.3ex]
(t_1^{\ell_1} (t_2^{\ell_2}))^\ell & \ifmodels t_1^{\ell_1} \wedge \ifmodels t_2^{\ell_2} &
	\forall \mathtt{FUN}(x,t_3^{\ell_3})\in\Gamma(\ell_1) \, . \ell_2 \direct x \wedge \ell_3 \direct \ell
								& \ell_1 \indirect \ell \\[0.3ex]
(\kw{if}(t_1^{\ell_1})\{ t_2^{\ell_2} \} \kw{else} \{ t_3^{\ell_3}\} )^{\ell_4} &
	\bigwedge_{i=1}^3 \ifmodels t_i^{\ell_i} &
	\ell_2 \direct \ell_4 \wedge \ell_3 \direct \ell_4 & \ell_1 \indirect \ell_4 \\[0.3ex]
(t_1, \rho)^{\ell_1} & \ifmodels {t_{1}}^{\ell_1} \wedge \bigwedge_{(x \mapsto t^\ell) \in \rho} \ifmodels t^\ell & - & - \\[0.3ex]
(\mm : t^{\ell_1})^{\ell_2} & \ifmodels t^{\ell_1} & \ell_1 \direct \ell_2 \wedge \mm \direct \ell_2 & -\\[0.3ex]
(t_1^{\ell_1}[t_2^{\ell_2}])^\ell & \ifmodels t_1^{\ell_1} \wedge \ifmodels t_2^{\ell_2} &
	\forall \mathtt{REC}(\ell') \in \Gamma(\ell_1) \, . \forall \ell'' \in \mathtt{proto}(\ell')_{\varrho} \, . & \ell_1 \indirect \ell \\
& &	\quad \forall s \, . \ell''.s \direct \ell & \ell_2 \indirect \ell \\[0.3ex]
\{ s_1 : t_1^{\ell_1} , \ldots , s_n : t_n^{\ell_n} \}^\ell & \bigwedge_{i=1}^n \ifmodels e_n &
	\exists \mathtt{REC}(\ell') \in \Gamma(\ell) \, . \forall i \, . \ell_i \direct \ell' . s_i & - \\[0.3ex]
(t_1^{\ell_1}[t_2^{\ell_2}] = t_3^{\ell_3})^\ell & \bigwedge_{i=1}^3 \ifmodels t_i^{\ell_i} &
	\ell_1 \direct \ell \wedge \forall \mathtt{REC}(\ell') \in \Gamma(\ell_1) \, . \forall s \, . \ell_3 \direct \ell'.s
		& \ell_2 \indirect \ell \\[0.3ex]
(\kw{del}\>t_1^{\ell_1}[t_2^{\ell_2}])^\ell & \ifmodels t_1^{\ell_1} \wedge \ifmodels t_2^{\ell_2} &
	\ell_1 \direct \ell & \ell_2 \indirect \ell \\[0.3ex]
(\kw{box}\>t^{\ell_1})^{\ell_2} & \ifmodels t^{\ell_1} & - & - \\[0.3ex]
(\kw{unbox}\>t^{\ell_1})^{\ell_2} & \ifmodels t^{\ell_1} &
	\forall \mathtt{BOX}(t'^{\ell'}) \in \Gamma(\ell_1) \, . \ell' \direct \ell_2 & \ell_1 \indirect \ell_2 \\[0.3ex]
(\kw{run}\>t^{\ell_1})^{\ell_2} & \ifmodels t^{\ell_1} &
	\forall \mathtt{BOX}(t'^{\ell'}) \in \Gamma(\ell_1) \, . \ell' \direct \ell_2 & \ell_1 \indirect \ell_2 \\[0.3ex]
(\kw{run}\>t^{\ell_1} \kw{in}\>\rho)^{\ell_2} & \ifmodels t^{\ell_1} \wedge \ifmodels \rho &
	\forall \mathtt{BOX}(t'^{\ell'}) \in \Gamma(\ell_1) \, . \ell' \direct \ell_2 & \ell_1 \indirect \ell_2 \\[0.3ex]
\hline
\end{array}
\]
}
\caption{Rules for generating information flow constraints}
\label{fig:ifrules}
\end{figure*}

We now prove the coherence of our information flow analysis with evaluation.
Like the corresponding proof for our 0CFA, this is lengthy and
technical, so we only sketch it here.
A mechanisation of the proof is available online.


\begin{lemma}[Reduction Preserves Satisfaction]
\label{l:ReductionPreserves}
If we have $\Gamma, \varrho, \flowsto \ifmodels t_1^{\ell_1}$ and also $t_1^{\ell_1} \red{n} t_2^{\ell_2}$,
then $\Gamma, \varrho, \flowsto \ifmodels t_2^{\ell_2}$.
Furthermore, $\ell_2 \flowsto^* \ell_1$.
\end{lemma}

\begin{proof}
By case analysis on the rules defining $\red{n}$.
\end{proof}


\begin{theorem}[Information Flow Coherence]
\label{l:EvaluationPreserves}
If we have $\Gamma, \varrho, \flowsto \ifmodels e_1$ and also $e_1 \ev{m} e_2$,
then $\Gamma, \varrho, \flowsto \ifmodels e_2$.
Furthermore, $\lbl(e_2) \flowsto^* \lbl(e_1)$.
\end{theorem}

\begin{proof}
\emph{Sketch:} Unfolding the definition of $\ev{m}$,
we let $e_1 = C^m_n \langle t_1^{\ell_1} \rangle$ and $e_2 = C^m_n \langle t_2^{\ell_2} \rangle$
with $t_1^{\ell_1} \red{n} t_2^{\ell_2}$.

Observe that $\Gamma, \varrho, \flowsto \ifmodels t_1^{\ell_1}$ and hence,
applying Lemma~\ref{l:ReductionPreserves}, $\Gamma, \varrho, \flowsto \ifmodels t_2^{\ell_2}$,
with $\ell_2 \flowsto^* \ell_1$.
Observe further that constraints generated by $C^m_n$ and the contents of its hole
interact only at that hole, labelled $\ell_2$ or $\ell_1$.
Thus, using $\ell_2 \flowsto^* \ell_1$, they must be satisfied in the conclusion,
giving $\Gamma, \varrho\, \flowsto \ifmodels C^m_n \langle t_2^{\ell_2} \rangle$
as required.

The claim that $\lbl(e_2) \flowsto^* \lbl(e_1)$ is trivial for all non-empty contexts,
as $\lbl(e_2) = \lbl(e_1)$.
For the empty context, it follows directly from the similar claim in Lemma~\ref{l:ReductionPreserves}.
\end{proof}

\begin{example}
Recall once more Example~\ref{ex:marked_fun}. Using the results of 0CFA,
our implementation generates the relations $\direct$ and $\indirect$
as depicted in Figure~\ref{fig:ifgraph}.
\begin{figure}
$
\begin{array}{cccccccccccccc}
4 & \direct & 5 & \direct & x & \direct & 0 & \sedirect & & \hspace{1em} & 7 & \direct & 8 & \direct y \\
\hsec & \nedirect & \isec & \sedirect & 3 & \indirect & 6 & \indirect & 9 & & \lsec & \nedirect \\
& & 1 & \direct & 2 & \nedirect
\end{array}
$
\caption{Information flow constraints for Example~\ref{ex:marked_fun}}
\label{fig:ifgraph}
\end{figure}

Setting $\flowsto = \direct \cup \indirect$,
we have $\hsec \flowsto^* 9$ and $\isec \flowsto^* 9$ and $\lsec \not{\!\!\!{\flowsto}^{*}} 9$.
As expected, this means the result (labelled $9$) has information flows from $\hsec$ and $\isec$,
but not $\lsec$.
\end{example}


\subsection{Information Flow Soundness}
\label{section:soundness}

\begin{theorem}[Information Flow Soundness]
Suppose $\Gamma, \varrho, \flowsto \ifmodels t^\ell$.
Then if $t^\ell \evstar{\anystage} v^{\ell'}$,
where $v$ is a stage-0 value composed only of markers and constants,
then $\erase{v}{M} = v$ where $M = \{ \mm \in Marker \mid \mm \flowsto^* \ell \}$.
\end{theorem}


\begin{proof}
First show that $\Gamma, \varrho, \flowsto \ifmodels v^{\ell'}$
with $\ell' \flowsto^* \ell$.
Argue by a simple induction over the derivation of $t^\ell \evstar{\anystage} v$.

\emph{Base case:} $\Gamma, \varrho, \flowsto \ifmodels t^\ell$
follows immediately from the theorem's premise.

\emph{Inductive step:}
Assume that $\Gamma, \varrho, \flowsto \ifmodels e_1$ and $\lbl(e_1) \flowsto^* \ell$,
with $e_1 \ev{\anystage} e_2$ the next step in the derivation.
Apply Theorem~\ref{l:EvaluationPreserves} to show that
$\Gamma, \varrho, \flowsto \ifmodels e_2$ and $\lbl(e_2) \flowsto^* \lbl(e_1)$;
hence $\lbl(e_2) \flowsto^* \ell$.

Now we have $\Gamma, \varrho, \flowsto \ifmodels v$
and $\ell' \flowsto^* \ell$.
Observe from the definition of $\erase{v}{M}$ that
if for every marker $\mm$ that occurs in $v$ we have $\mm \in M$,
then $\erase{v}{M} = v$.

But $v$ is a value composed only of markers and constants,
so for every marker $\mm $ that occurs in $v$
(by examination of the $\ifmodels$ constraint rules)
it must be the case that $\mm  \flowsto^* \ell'$.
Thus, as $\ell' \flowsto^* \ell$, $\mm  \flowsto^* \ell$.
Hence, from the definition of $M$, $\mm  \in M$.
So it is indeed true that $\erase{v}{M} = v$.
\end{proof}

\subsubsection*{Relationship with Noninterference}


Our information flow analysis can be used to verify the security property noninterference.
Noninterference asserts that the values of any ``high-security'' inputs must not affect
the values of any ``low-security'' outputs.
In order for this assertion to be meaningful, we must have notions of input, output
and high- and low-security levels.

For example, assume elements of $\Marker$ represent different levels of security,
such as $\lsec$ for low security and $\hsec$ for high security.
For input, assume two relations $\lowin$ and $\highin$, which take an expression and set the values of
low and high inputs respectively. 
For low-security output, just take the value to which an expression evaluates.

Say that expression $t^\ell$ satisfies noninterference analysis if $\Gamma, \varrho, \flowsto \ifmodels t^\ell$
and $\hsec \not{\!\!\!{\flowsto}^{*}} \ell$.
Further, require that $\lowin$ and $\highin$ satisfy the following conditions:
\[
\begin{array}{rcl}
\Gamma, \varrho, \flowsto \ifmodels t^\ell \wedge t \lowin t' &
	\implies & \Gamma, \varrho, \flowsto \ifmodels t'^\ell \\
\Gamma, \varrho, \flowsto \ifmodels t^\ell \wedge t \highin t' &
	\implies & \Gamma, \varrho, \flowsto \ifmodels t'^\ell \\
\Gamma, \varrho, \flowsto \ifmodels t^\ell \wedge t \highin t' &
	\implies & \erase{t}{\lowset} = \erase{t'}{\lowset}
\end{array}
\]
\emph{Claim:} If $t^\ell$ satisfies noninterference analysis, then in the following situation:
\[
\begin{array}{ccccccc}
t^\ell \lowin t'^\ell & \hspace{1em} &
t'^\ell \highin t_1^\ell & \hspace{1em} &
t'^\ell \highin t_2^\ell & \hspace{1em} &
t_1^\ell \evstar{\anystage} u^{\ell'}
\end{array}
\]
where $u^{\ell'}$ is a value composed only of markers and constants,
it follows that $t_2^\ell \evstar{\anystage} u^{\ell'}$.
That is, the low output $u$ is independent from
the values of the high inputs for $t$ selected using $\highin$.

\emph{Proof:}
By the condition on $\lowin$, observe we have $\Gamma, \varrho, \flowsto \ifmodels t'$.
By the first condition on $\highin$, it then follows that $\Gamma, \varrho, \flowsto \ifmodels t_1^\ell$
and $\Gamma, \varrho, \flowsto \ifmodels t_2^\ell$.
As $\hsec \not{\!\!\!{\flowsto}^{*}} \ell$, by soundness of information flow,
we have $u = \erase{u}{\lowset}$.
So using stability, we get $\erase{t_1}{\lowset} \evstar{\anystage} u$.
But, by the second condition on $\highin$, we have
$\erase{t'}{\lowset} = \erase{t_1}{\lowset} = \erase{t_2}{\lowset}$.
So $\erase{t_2}{\lowset} \evstar{\anystage} u$.
Then by monotonicity, $t_2 \evstar{\anystage} u$.

The conditions on $\lowin$ and $\highin$ seem reasonable.
As an example, $\lowin$ and $\highin$ that can only replace constants marked as $\lsec$ and $\hsec$ respectively
and can only replace them with constants of the same type (integer, boolean or string) satisfy these conditions.

\section{Evaluation}\label{sec:implementation}

We have implemented our analysis in OCaml and tested it on a range of
examples.
The source code for our analysis tool and the examples are available online.
We now present some of these examples.

For each example, we list the markers on which our analysis says the result
may depend.
To improve readability, we write $\kw{let}\>x = v \kw{in}\>e$
as a shorthand for $\kw{fun}(x) \{ e \} v$.
Our implementation extends \slamjs (and its analysis) as presented in this paper with
primitive arithmetic, equality and \kw{typeof} operators,
which we use in some of our examples.
It can also handle mutable references in the style of \lambdajs\ and a subset of actual JavaScript syntax.
Many of our examples are inspired by patterns of $\kw{eval}$ usage
common in Web applications, as surveyed by Richards et
al.~\cite{richards11} and discussed by Jensen et al.~\cite{jensen12}.
\vspace{0.5\baselineskip}
%
\begin{examplelist}
\emph{Depends on}: $\highsec, \lsec$. \\
$\kw{if}(\hsec : \kw{true}) \{ \lsec : \kw{false} \} \kw{else} \{ 1\}$
\\
We begin with a classic example where branching on a value introduces an indirect flow from it.
As our analysis does not track specific boolean values,
it would give the same result if the branch were on $(\hsec : \kw{false})$.
We could resolve this imprecision by extending our abstract value domain with abstract values for $\kw{true}$ and $\kw{false}$.
\end{examplelist}
\vspace{0.5\baselineskip}
\begin{examplelist}
\emph{Depends on}: $\hsec, \isec, \lsec$. \emph{Depends (improved)}: $\hsec, \lsec$. \\
$\kw{let}\> \vn{ctrue} = \kw{fun} (x) \{ \kw{fun} (y) \{ x \} \} \kw{in} \\
\kw{let}\> \vn{cif} = \kw{fun} (x) \{ \kw{fun} (y) \{ \kw{fun} (z) \{ (x(y))(z) \} \} \} \kw{in} \\
((\vn{cif}(\hsec : \vn{ctrue}))(\lsec:\kw{false}))(\isec:1) )$
\\
Conversely, if we present the previous example using the standard Church-encodings of $\kw{if}$ and $\kw{true}$
as functions, our analysis is precise enough to determine that the result does not depend on $\isec$.
Note that we need the improved analysis to distinguish
the bindings of $x$ and $y$ in $\vn{ctrue}$ and $\vn{cif}$.
\end{examplelist}
\vspace{0.5\baselineskip}
\begin{examplelist}
\emph{Depends on}: $\lsec$. \\
$\kw{let}\>x = \kw{if} (\kw{true}) \{ \kw{box}\>f \} \kw{else} \{ \kw{box}\>g \} \kw{in} \\
\kw{let}\>f = \kw{fun} (y) \{ 1 \} \kw{in} \\
\kw{let}\>g = \kw{fun} (z) \{ \lsec : \kw{true} \} \kw{in} \\
\kw{run}\>(\kw{box}\>((\kw{unbox}\>x)(\hsec : \kw{undef})))$ \\
This example is modelled on the following JavaScript usage pattern~\cite{jensen12}: \\
\texttt{if (...) x = "f"; else} \\
\texttt{x = "g"; eval(x + "()");} \\
$f$ and $g$ are bound to functions; $x$ is set to a code value of either $f$ or $g$;
a function argument is added to the code value and the result executed.
In this example, both $f$ and $g$ ignore their argument $(\hsec : \kw{undef})$,
so the result does not depend on $\hsec$; our analysis correctly identifies this.
\end{examplelist}
\vspace{0.5\baselineskip}
\begin{examplelist}
\emph{Depends on}: $\hsec, \lsec$.
\emph{Depends (improved)}: $\lsec$. \\
$\kw{let}\>c = \kw{box}\>x \kw{in} \\
\kw{let}\>x = \lsec : 1 \kw{in} \\
\kw{let}\>\vn{eval} = \kw{fun} (b) \{ \kw{run}\>b \} \kw{in} \\
\kw{let}\>x = \hsec : 2 \kw{in} \\
\vn{eval}(c)$ \\
JavaScript programmers sometimes use $\kw{eval}$ to execute code within a different scope.
\slamjs does not aim to emulate all the quirks of $\kw{eval}$,
but scoping of staged code can still have interesting behaviour,
as shown in this example.
In the scope of the definition of the function bound to $\vn{eval}$,
$x$ is $1$. So when it evaluates the code value $c$, which contains just the
variable $x$, this is the value it returns; note that $x$ was not bound at
all where $c$ was defined.
The second binding of $x$ is unused; our analysis correctly determines this.
\end{examplelist}
\vspace{0.5\baselineskip}
\begin{examplelist}
\emph{Depends on}: $\hsec, \isec, \lsec$. \\
$\kw{let}\>i = \isec : \\
 \{ \proto : \kw{null} , \quot x \quot : (\hsec : 1) , \quot y \quot : (\lsec : 2) \} \kw{in} \\
\kw{let}\>s = \kw{fun} (\vn{id}) \{ 
\kw{let}\>f = \kw{box}\>(i[\kw{unbox}\>\vn{id}]) \kw{in}\>
    \kw{run}\>f
  \} \kw{in} \\
  s(\kw{box}\>{\quot}y{\quot})$ \\
Some JavaScript programmers use $\kw{eval}$ to construct variable names,
as in
(\texttt{var n = 5; eval ({\quot}f\_{\quot} + n);})
to access \texttt{f\_5}.
We cannot express this directly in \slamjs because there are no facilities to manipulate variable names.
Another common practice is to use $\kw{eval}$ to access object properties,
often because of the programmer's ignorance of JavaScript's indirect object field access syntax;
this example models that practice in \slamjs.
Because our analysis does not model the values strings may take,
its handling of field reads and writes is rather coarse,
so it cannot tell the result will not depend on $\hsec$;
this could be addressed refining our abstract value domain.
\end{examplelist}
\vspace{0.5\baselineskip}
\begin{examplelist}
\emph{Depends on}: $\hsec$. \\
$\kw{let}\>\vn{fst} = \kw{fun} (x) \{ \kw{fun} (y) \{x\}\} \kw{in} \\
\kw{let}\>f = \kw{if} (\kw{false}) \{\vn{fst}\} \kw{else} \{\kw{box} \vn{fst}\} \kw{in} \\
\kw{let}\>x = (\hsec:1) \kw{in} \\
\kw{let}\>y = (\lsec:\kw{true}) \kw{in} \\
    \kw{if} (\kw{typeof}\>f  = \mathtt{{\quot}function{\quot}}) 
      \{ (f(x))(y) \} \\
    \kw{else} 
      \{ \kw{run}\>(\kw{box}\>(((\kw{unbox}\>f)(x))(y)) ) \} $ \\
This example models the JavaScript usage pattern: \\
\texttt{if (f instanceof Function) f(x);} \\
\texttt{else eval (f + {\quot}(x){\quot});} \\
which may arise when using $\kw{eval}$ to emulate higher-order functions.
Here, our analysis shows the same precision on a boxed value representing a function
as when dealing with a real function.
\end{examplelist}
\vspace{0.5\baselineskip}
\begin{examplelist}
\emph{Depends on:} $\hsec, \lsec$.
\emph{Depends (improved)}: $\lsec$. \\
$\kw{let}\>\vn{pair} = \kw{fun} (x) \{ \kw{fun} (y) \{ \kw{fun} (z) \{ \kw{run}\>z \} \} \} \kw{in} \\
  \kw{let}\>\vn{fst} = \kw{fun} (z) \{ z(\kw{box}\>x) \} \kw{in} \\
  \kw{let}\>\vn{snd} = \kw{fun} (z) \{ z(\kw{box}\>y) \} \kw{in} \\
  \kw{let}\>\vn{bp} = \kw{box}\>((\vn{pair}(\lsec : (\kw{box}\>(1))))(\hsec : (\kw{box}\>(\kw{true})))) \kw{in} \\
  \kw{let}\>\vn{boxfst} = \kw{box}\>((\vn{fst})(\kw{unbox}\>\vn{bp})) \kw{in} \\
  \kw{run}\>(\kw{run}\>(\vn{boxfst}))$ \\
Most examples of staged metaprogramming in the literature do not use more than
one level of staging.
This example, which pairs and unpairs two values in a
rather roundabout way, illustrates that we can handle higher levels too.
\end{examplelist}
\vspace{0.5\baselineskip}
\begin{examplelist}
\emph{Depends on}: $\hsec$. \\
$\kw{fun}(n) \{(\kw{fun} (x) \{ (x(x))(n) \} ) \\
(\kw{fun} (x) \{ \kw{fun} (y) \{ \kw{if} (y = 0) \{ \kw{true} \} \kw{else} \{(x(x))(y-1) \} \} \}) \\
\} (\hsec : 5)$ \\
This program loops $n$ times (where $n$ is $(\hsec : 5)$ in this instance) before returning $\kw{true}$.
In this sense, the result is independent of $n$:
if $n$ were a high-security input and the output low, the program would satisfy noninterference,
although the duration of execution may leak information about $n$.
However, $n$ must be examined in order to execute the program, so there is an information flow from
$n$ to the result, in the sense captured by our labelled semantics.
That is, no noninterference analysis based on a sound over-approximation of the behaviour
of such a labelled semantics could ever show the program to be noninterfering.
\end{examplelist}
\vspace{0.5\baselineskip}
\begin{examplelist} \label{ex:choi1} \emph{Depends on:} $\lsec$. \\
$
\kw{let}\> \vn{fst} = \kw{fun} (x) \{ \kw{fun} (y) \{ x \} \} \kw{in} \\
\kw{let}\> a = \kw{box}\> x \kw{in} \\
\kw{let}\> b = \kw{box}\> (\kw{fun} (x) \{ \kw{fun}(y) \{ \vn{fst}(\kw{unbox} a)(y) \} \}) \kw{in} \\
(\kw{run} b)(\lsec : 1)(\hsec : 2) \\
$
This program, based on an example from Choi et al.~\cite{choi11},
splices a variable name into a code template to produce code that
takes two arguments and returns the first.
Our analysis correctly determines that the result depends only on the first.
\end{examplelist}
\vspace{0.5\baselineskip}
\begin{examplelist} \label{ex:choi2} \emph{Depends on:} $\lsec, \hsec$. \\
$
\kw{let}\> \vn{fst} = \kw{fun} (x) \{ \kw{fun} (y) \{ x \} \} \kw{in} \\
\kw{let}\> a = \kw{fun} (p) \{p[\texttt{"x"}]\} \kw{in} \\
\kw{let}\> b = (\kw{fun} (h) \{ \kw{fun} (p) \{ \kw{fun} (x) \{ \kw{fun} (y) \{ \\
\> \vn{fst}(h ((p[\texttt{"x"}]=x)[\texttt{"y"}]=y))(y) \}\}\}\})(a) \kw{in} \\
b(\{\texttt{"\_\_proto\_\_"} : \kw{null} \})(\lsec : 1)(\hsec : 2) \\
$
By applying Choi et al.'s unstaging translation to the core of the previous example, we
obtain this unstaged one. Note that while the result of the program is the
same, we lose precision by analysing this version instead of
working directly on the staged version.
\end{examplelist}
\vspace{0.5\baselineskip}
\begin{examplelist} \label{ex:choi3} \emph{Depends on:} $\lsec, \hsec$. \emph{Depends (improved):} $\lsec$.\\
$
\kw{let}\> blank = \kw{fun} (get) \{ get(\kw{null})(\kw{null}) \} \kw{in} \\
\kw{let}\> getx = \kw{fun} (x) \{ \kw{fun} (y) \{ x \} \} \kw{in} \\
\kw{let}\> gety = \kw{fun} (x) \{ \kw{fun} (y) \{ y \} \} \kw{in} \\
\kw{let}\> setx = \kw{fun} (env) \{ \kw{fun} (newx) \{ \\
\> \kw{fun} (get) \{ get(newx)(env(gety)) \} \} \} \kw{in} \\
\kw{let}\> sety = \kw{fun} (env) \{ \kw{fun} (newy) \{ \\
\> \kw{fun} (get) \{ get(env(getx))(newy) \} \} \} \kw{in} \\
\kw{let}\> \vn{fst} = \kw{fun} (x) \{ \kw{fun} (y) \{ x \} \} \kw{in} \\
\kw{let}\> a = \kw{fun} (p) \{p(getx)\} \kw{in} \\
\kw{let}\> b = (\kw{fun} (h) \{ \kw{fun} (p) \{ \kw{fun} (x) \{ \kw{fun} (y) \{ \\
\> \vn{fst}(h (sety(setx(p)(x))(y)))(y) \}\}\}\})(a) \kw{in} \\
b(blank)(\lsec : 1)(\hsec : 2)
$ \\
Here we have applied the unstaging translation, as in the previous example,
but using higher order functions to encode environments instead of records.
In this case, we can recover the lost precision,
but at the cost of an $\mathcal{O}(n^2)$ increase in the size of the source program,
making the combined analysis $\mathcal{O}(n^6)$ instead of $\mathcal{O}(n^3)$~\cite{DBLP:conf/lics/HeintzeM97}.
\end{examplelist}

\section{Related Work}
\label{sec:related}
\subsection{From \slamjs to JavaScript Applications}\label{sec:applications}

The application that guided our work is
information flow analysis for JavaScript in Web applications.
We now consider some of the features of this scenario that we have not addressed
and how they have been handled by others.
We claim that most of the problems have been addressed,
although combining them into a single analysis system would require further effort.

\emph{Handling of Primitive Datatypes}
As demonstrated in some of our examples,
our analysis models its primitive datatypes (such as strings and booleans) very coarsely;
our abstract domains are too simple.
Fortunately, more refined abstractions for these datatypes have been well-studied~\cite{DBLP:conf/aplas/ChoiLKD06}.

\emph{Imperative Control Flow and Exceptions}
JavaScript has several features not found in \slamjs,
including typical imperative control flow features (such as $\kw{for}$ loops)
and exceptions, but there are CFA-style analyses for JavaScript.
Perhaps most notable is the recent CFA2 analysis~\cite{vardoulakis11},
which was developed for JavaScript
and features significantly better analysis of higher order flow control.

\emph{JavaScript Semantics}
A bigger problem in producing a sound analysis of JavaScript
is the complexity and quaintness of its semantics~\cite{DBLP:conf/aplas/MaffeisMT08}.
Guha et al. attempt to simplify this problem by
producing a much simpler ``core calculus'' for JavaScript called \lambdajs\ and a transformation from JavaScript into \lambdajs~\cite{guha10}.
They have mechanised various proofs about their language in Coq.
As Web applications execute in the context of a webpage in a browser,
an analysis must also model how a webpage interacts with code via the DOM.

\emph{Code Strings vs Staged Code}
Perhaps the most relevant difference between JavaScript and \slamjs
is our metaprogramming constructs:
JavaScript $\kw{eval}$ runs on strings,
while, in an effort to develop a more principled analysis,
our staged metaprogramming follows the tradition of Lisp quotations.
To analyse uses of $\kw{eval}$ with our techniques,
we would need a sound transformation into staged metaprogramming.
Jensen et al. use the result of a string analysis
produced by the tool TAJS to replace certain uses of $\kw{eval}$
with unstaged code where it is safe to do so~\cite{jensen12};
the transformed program is then fed back into the analysis tool.
We propose to handle a wider range of use cases
with the more general approach of
transforming $\kw{eval}$ on strings into staged code
and then analysing the staged code.

\emph{Reactive Systems}
A practical Web application is not simply a program that take inputs,
runs once, then gives output:
it may interleave input and output throughout its execution, which might not terminate.
Bohannon et al. consider the consequences of this
for information security in their work on reactive noninterference~\cite{DBLP:conf/ccs/BohannonPSWZ09}.

\emph{Infrastructural Issues}
In applying an information flow analysis to a Web application,
several infrastructural issues need to be addressed.
Would the code be analysed before being published by on a webserver,
in the browser running it or by some proxy in between?
Will the entire code be available in advance, or must it be analysed in fragments~\cite{chugh09}?
Who would set the security policies that the analysis should enforce?
Li and Zdancewic argue that noninterference alone is too strict a policy to enforce
and that a practical policy must allow for limited declassification~\cite{DBLP:conf/popl/LiZ05}.

\subsection{Information Flow Analysis}\label{sec:related-if}

Early work on information flow security focused on monitoring
program execution, dynamically marking variables to indicate their
level of confidentiality~\cite{DBLP:journals/cj/Fenton74}.
However, the study of static analysis for information flow security
can essentially be traced back to Denning, who introduced a lattice
model for secure information flow and critically considered both
direct and indirect flows~\cite{DBLP:journals/cacm/Denning76}.
Denning and Denning developed a simple static information flow
analysis that rejected programs with flows violating a security
policy~\cite{DBLP:journals/cacm/DenningD77}.

\emph{Noninterference}
Goguen and Meseguer introduced the idea of
noninterference~\cite{goguen82} (the inability of the actions of one
party, or equivalently data at one level, to influence those of
another) as a way of specifying security policies, including
enforcement of information flow security.
Noninterference and information flow security became almost synonymous,
although Pottier and Conchon were careful to emphasise the distinction
between the two~\cite{DBLP:conf/icfp/PottierC00}.

\emph{Security Type Systems}
Security type systems became a common way of enforcing noninterference
policies and proving the correctness of noninterference analyses,
progressing from a reformulation of Denning and Denning's
analysis~\cite{DBLP:journals/jcs/VolpanoIS96} to Simonet and Pottier's
type system for ML~\cite{DBLP:journals/toplas/PottierS03}.
Unfortunately, the requirement that the program analysed follow
a strict type discipline makes it impractical to apply these ideas to
dynamically typed languages such as JavaScript.
Perhaps as a consequence, information flow in untyped and dynamically typed
languages is relatively poorly understood.

\emph{Dynamic Analyses}
Dynamic information flow analysis circumvents the need for a type
system or other static analysis by tracking information flow during
program execution, and enforcing security policies by aborting program
execution if an undesired flow is detected; examples of such
analyses for JavaScript are presented by Just et al.~\cite{just11}
and Hedin and Sabelfeld~\cite{DBLP:conf/csfw/HedinS12}.
Indeed, the problems they address and their motivations are very
similar to ours, but our methods are very different.

\emph{Dynamic vs Static}
A dynamic analysis only observes one program run at a time, so dynamic
code generation is easy to handle. However, care has to be taken to
track indirect information flow due to code that was \emph{not}
executed in the observed run. Strategies to achieve this include, for
instance, the \emph{no-sensitive upgrade}
check~\cite{zdancewic-thesis}, which aborts execution if a public
variable is assigned in code that is control dependent on private
data. As a rule, however, such strategies are fairly coarse and could
potentially abort many innocuous executions; thus it is commonly held
that static analyses are superior to dynamic ones in their treatment
of indirect flows~\cite{DBLP:journals/jsac/SabelfeldM03},
although there has been a resurgence of interest
in dynamic analyses~\cite{DBLP:conf/ershov/SabelfeldR09}.

\emph{Hybrid Approaches}
As a compromise, Chugh et al.~\cite{chugh09} propose extending a
static information flow analysis with a dynamic component that
performs additional checks at runtime when dynamically generated code
becomes available. The static part of their analysis is similar to ours
(minus staging), although they do not formally state or prove its soundness.
Their study of JavaScript on popular websites suggests
the static part is precise enough to be useful.
Because the additional checks on dynamically generated code occur at
runtime, they must necessarily be quick and simple to avoid performance degradation.
Consequently, these checks are limited to purely syntactic isolation properties,
with a corresponding loss of precision.
Our fully static analysis does not suffer from these limitations.

Going in the other direction, Austin and
Flanagan~\cite{austin12} have proposed \emph{faceted execution}, a
form of dynamic analysis that explores different execution paths and
can thus recover some of the advantages of a static analysis.

\subsection{Static Analysis of Staged Metaprogramming}
\label{sec:related-meta}

Many different approaches to staged metaprogramming have been
proposed. Our language's staging constructs are modelled after
the language $\lambda_{\mathcal{S}}$ of Choi et al.~\cite{choi11}.
However, our semantics of variable capture are different.
For example, we allow the program
$(\kw{fun}(x) \{ \kw{run}\>(\kw{box}\>x) \} (1))$,
which behaves much like this JavaScript program:
$\texttt{(function (x) \{return eval("x")\})(1);}$

Control flow analysis for a two-staged language has been investigated
by Kim et al.~\cite{kim09}. Their approach is based on abstract interpretation,
putting particular emphasis on inferring an over-approximation
of all possible pieces of code to which a code quotation may evaluate.
This information is not explicitly computed by our
analysis, so it is quite possible that their analysis is more precise
than ours. However it does not seem to have been implemented yet.

Choi et al.~\cite{choi11} propose a more general framework for static
analysis of multi-staged programs, which is based on an unstaging
translation that replaces staging constructs with function
abstractions and applications. Under certain conditions, analysis
results for the unstaged program can then be translated back to its
staged version.

There are some limitations to their work.
Most significantly, many interesting programs, such as the one mentioned earlier,
are not valid in $\lambda_{\mathcal{S}}$ and hence
cannot be unstaged using their translation;
this limits its applicability to JavaScript.
Furthermore, as shown in Examples~\ref{ex:choi1}--\ref{ex:choi3},
the precision of the resulting combined analysis is highly sensitive to the
target language encoding used in the translation and the behaviour of the target language analysis.
While their approach is useful as a quick way of adding staging to an
existing language and analysis, we argue that staging constructs are
sufficiently important and complex that
we should aim to analyse them directly.

Inoue and Taha~\cite{DBLP:conf/esop/InoueT12} consider the problem of
reasoning about staged programs; in particular, they identify
equivalences that fail to hold in the presence of staging, and develop
a notion of bisimulation that can be used to prove extensionality of
function abstractions, and work around some of the failing
equivalences. Their language differs from ours in that it avoids name
capture.

\section{Conclusions}\label{sec:conclusions}

We have presented a fully static information flow analysis based on 0CFA
for a dynamically typed language with staged metaprogramming,
implemented it and formally proved its soundness.
We believe our approach is transferrable to other CFA-style analyses
and applicable to JavaScript.

Progressing from here, there are three obvious lines of work.
The first is to improve the precision of the analysis by applying its ideas
to CFA2 or using results from abstract interpretation.
The second is to extend the language to handle more features,
such as imperative control flow and exceptions.
The third and most important is to apply string analysis techniques to
produce a sound transformation from a language with $\kw{eval}$ on code
strings to a language with staged code values.

All the pieces are now in place for an
interesting, sound and principled analysis of JavaScript with $\kw{eval}$,
but it will take significant effort to bring them together.





\bibliographystyle{IEEEtran}
\bibliography{IEEEabrv,refs}
%

\clearpage


\appendix[\slamjs Semantics Definitions]
\label{a:sem}

Some of the less interesting definitions of \slamjs semantics are given here
in full.

\begin{figure}[h!]
\[
\begin{array}{rclcl}
C^m_n & ::=   & [\,] & \in & \Ctxt^n_n \\
      & \mid & (\{\overline{s:v^m}, s:C^m_n, \overline{s:e}\}) & \in & \Ctxt^m_n \\
      & \mid & (\kw{fun}(x)\{C^{m+1}_n\}) & \in & \Ctxt^{m+1}_n \\
      & \mid & (C^m_n(e)) & \in & \Ctxt^m_n \\
      & \mid & (v^m(C^m_n)) & \in & \Ctxt^m_n \\
      & \mid & (\kw{box}\>C^{m+1}_n) & \in & \Ctxt^m_n \\
      & \mid & (\kw{unbox}\>C^m_n) & \in & \Ctxt^{m+1}_n \\
      & \mid & (\kw{run}\>C^m_n) & \in & \Ctxt^m_n \\
      & \mid & (\kw{if}(C^m_n)\{e\}\kw{else}\{e\}) & \in & \Ctxt^m_n \\
      & \mid & (\kw{if}(v^{m+1})\{C^{m+1}_n\}\kw{else}\{e\}) & \in & \Ctxt^{m+1}_n \\
      & \mid & (\kw{if}(v^{m+1})\{v^{m+1}\}\kw{else}\{C^{m+1}_n\}) & \in & \Ctxt^{m+1}_n \\
      & \mid & (C^m_n[e]) & \in & \Ctxt^m_n \\
      & \mid & (v^m[C^m_n]) & \in & \Ctxt^m_n \\
      & \mid & (C^m_n[e]=e) & \in & \Ctxt^m_n \\
      & \mid & (v^m[C^m_n]=e) & \in & \Ctxt^m_n \\
      & \mid & (v^m[v^m]=C^m_n) & \in & \Ctxt^m_n \\
      & \mid & (\kw{del}\>C^m_n[e]) & \in & \Ctxt^m_n \\
      & \mid & (\kw{del}\>v^m[C^m_n]) & \in & \Ctxt^m_n \\
      & \mid & (\kw{run}\>C^m_n\>\kw{in}\>\rho) & \in & \Ctxt^m_n
\end{array}
\]
\caption{Evaluation contexts}\label{fig:contexts}
\end{figure}

\begin{figure}[h!]
\[
\begin{array}{rcl}
(k,\rho) & \red{n} & k \\
(\{\overline{s:e}\},\rho) & \red{n} & \{\overline{s:(e,\rho)}\} \\
(x,\rho) & \red{n+1} & x \\
(\kw{fun}(x)\{e\},\rho) & \red{n+1} & (\kw{fun}(x)\{(e,\rho)\}) \\
(e_1 (e_2),\rho) & \red{n} & ((e_1,\rho)((e_2,\rho))) \\
(\kw{box}\>e,\rho) & \red{n} & (\kw{box}\>(e,\rho)) \\
(\kw{unbox}\>e,\rho) & \red{n} & (\kw{unbox}\>(e,\rho)) \\
(\kw{run}\>e,\rho) & \red{0} & (\kw{run}\>(e,\rho)\>\kw{in}\>\rho) \\
(\kw{run}\>e,\rho) & \red{n+1} & (\kw{run}\>(e,\rho)) \\
(\kw{if}(e_1)\{e_2\}\kw{else}\{e_3\},\rho) & \red{n} & (\kw{if}((e_1,\rho))\{(e_2,\rho)\}\kw{else}\{(e_3,\rho)\}) \\
(e_1[e_2],\rho) & \red{n} & ((e_1,\rho)[(e_2,\rho)]) \\
(e_1[e_2]=e_3,\rho) & \red{n} & ((e_1,\rho)[(e_2,\rho)]=(e_3,\rho)) \\
(\kw{del}\>e_1[e_2],\rho) & \red{n} & (\kw{del}\>(e_1,\rho)[(e_2,\rho)]) \\
\end{array}
\]
\caption{Environment propagation rules}\label{fig:environment propagation rules}
\end{figure}

\begin{figure}[h!]
\[
\begin{array}{rcll}
\erase{\blank}{M} & = & \blank \\

\erase{k}{M} & = & k \\
\erase{\{\overline{s:e}\}}{M} & = & \{\overline{s:\erase{e}{M}}\} \\
\erase{x}{M} & = & x \\
\erase{\kw{fun}(x) \{ e \}}{M} & = & \kw{fun}(x) \{ \erase{e}{M} \} \\
\erase{e_1(e_2)}{M} & = & \erase{e_1}{M}(\erase{e_2}{M}) \\
\erase{\kw{box}\>e}{M} & = & \kw{box}\>\erase{e}{M} \\
\erase{\kw{unbox}\>e}{M} & = & \kw{unbox}\>\erase{e}{M} \\
\erase{\kw{run}\>e}{M} & = & \kw{run}\>\erase{e}{M} \\
\erase{ \kw{if} (e_1) \{ e_2 \} \kw{else} \{ e_3 \} }{M} & = &
	\kw{if} (\erase{e_1}{M}) \{ \erase{e_2}{M} \} \kw{else} \{ \erase{e_3}{M} \} \\
\erase{e_1[e_2]}{M} & = & \erase{e_1}{M}[\erase{e_2}{M}] \\
\erase{e_1[e_2]=e_3}{M} & = & \erase{e_1}{M}[\erase{e_2}{M}] = \erase{e_3}{M} \\
\erase{\kw{del}\>e_1[e_2]}{M} & = & \kw{del}\>\erase{e_1}{M}[\erase{e_2}{M}] \\

\erase{(e, \rho)}{M} & = & (\erase{e}{M}, \erase{\rho}{M} ) \\
\erase{\kw{run}\>e \kw{in}\>\rho}{M} & = & \kw{run}\>\erase{e}{M} \kw{in}\>\erase{\rho}{M} \\
\erase{\rho}{M} (x) & = & \erase{\rho(x)}{M} \\

\erase{\mm : e}{M} & = & \mm  : \erase{e}{M} & \hspace{-6em} \text{if $\mm  \in M$} \\
\erase{\mm : e}{M} & = & \blank & \hspace{-6em} \text{if $\mm  \notin M$} \\
\end{array}
\]
\caption{Definition of $\erase{e}{M}$, the $M$-erasure of $e$}
\end{figure}

\clearpage
\appendix[0CFA for \slamjs]\label{app:cfa}

\subsection{Labelled Semantics}

We extend the syntax of \slamjs with labels to indicate program points.
The labels have no effect on the result of computation,
but are used to track which values may occur at which points.
Consequently, it is important for the soundness of the corresponding analysis
that the semantics correctly tracks labels.

We reformulate the syntax of \slamjs to distinguish between terms
(expressions in the unlabelled semantics) and expressions (which are labelled terms):

\[
\begin{array}{llcl}
\text{Expressions} & e & ::= & t^\ell \\
\text{Terms} & t & ::= & k \mid \{\overline{s:e}\} \mid x \mid \kw{fun}(x)\{e\} \mid e(e) \\
& & \mid & \kw{box}\>e \mid \kw{unbox}\>e  \mid \kw{run}\>e \\
             &   & \mid & \kw{if}(e)\{e\}\kw{else}\{e\} \mid e[e] \mid e[e]=e \\
& & \mid & \kw{del}\>e[e] \mid (t,\rho) \mid \kw{run}\>e\>\kw{in}\>\rho \\
\end{array}
\]

Values remain expressions, so they include labels at the outer level.
For example, $k^\ell$ is a value, rather than $k$.
Contexts other than the empty context also gain labels at the outer level,
so we have $(C^m_n \ang{e})^\ell$ rather than $(C^m_n \ang{e})$.

The labelling of the reduction rules is a little more complicated,
so we list them in full in
Figure~\ref{fig:labelled_environment_propagation_rules},~\ref{fig:labelled_proper_reduction_rules}~and~\ref{fig:labelled_lifts}.
For an expression $e = t^\ell$, we write $e^{\ell'}$ as a shorthand for $t^{\ell'}$
and $(e,\rho)^{\ell'}$ for $(t,\rho)^{\ell'}$.
Note that we use this in the rules \rlname{Lookup}, \rlname{Unbox}, \rlname{Run} and \rlname{Read1}.

\begin{figure*}
\[
\begin{array}{rcl}
(k,\rho)^\ell & \red{n} & k^\ell \\
(\{\overline{s:t^\ell}\},\rho)^{\ell'} & \red{n} & \{\overline{s:(t,\rho)^\ell}\}^{\ell'} \\
(x,\rho)^\ell & \red{n+1} & x^\ell \\
(\kw{fun}(x)\{t^{\ell}\},\rho)^{\ell'} & \red{n+1} & (\kw{fun}(x)\{(t,\rho)^{\ell}\})^{\ell'} \\
(t_1^{\ell_1}(t_2^{\ell_2}),\rho)^\ell & \red{n} & ((t_1,\rho)^{\ell_1}((t_2,\rho)^{\ell_2}))^\ell \\
(\kw{box}\>t^\ell,\rho)^{\ell'} & \red{n} & (\kw{box}\>(t,\rho)^\ell)^{\ell'} \\
(\kw{unbox}\>t^\ell,\rho)^{\ell'} & \red{n} & (\kw{unbox}\>(t,\rho)^\ell)^{\ell'} \\
(\kw{run}\>t^\ell,\rho)^{\ell'} & \red{0} & (\kw{run}\>(t,\rho)^\ell\>\kw{in}\>\rho)^{\ell'} \\
(\kw{run}\>t^\ell,\rho)^{\ell'} & \red{n+1} & (\kw{run}\>(t,\rho)^\ell)^{\ell'} \\
(\kw{if}(t_1^{\ell_1})\{t_2^{\ell_2}\}\kw{else}\{t_3^{\ell_3}\},\rho)^{\ell_0} & \red{n} & (\kw{if}((t_1,\rho)^{\ell_1})\{(t_2,\rho)^{\ell_2}\}\kw{else}\{(t_3,\rho)^{\ell_3}\})^{\ell_0} \\
(t_1^{\ell_1}[t_2^{\ell_2}],\rho)^{\ell_0} & \red{n} & ((t_1,\rho)^{\ell_1}[(t_2,\rho)^{\ell_2}])^{\ell_0} \\
(t_1^{\ell_1}[t_2^{\ell_2}]=t_3^{\ell_3},\rho)^{\ell_0} & \red{n} & ((t_1,\rho)^{\ell_1}[(t_2,\rho)^{\ell_2}]=(t_3,\rho)^{\ell_3})^{\ell_0} \\
(\kw{del}\>t_1^{\ell_1}[t_2^{\ell_2}],\rho)^{\ell_0} & \red{n} & (\kw{del}\>(t_1,\rho)^{\ell_1}[(t_2,\rho)^{\ell_2}])^{\ell_0} \\
(\mm : t_1^{\ell_1},\rho)^\ell & \red{n} & (\mm : (t_1,\rho)^{\ell_1})^\ell
\end{array}
\]
\caption{Labelled environment propagation rules}\label{fig:labelled_environment_propagation_rules}
\end{figure*}

\begin{figure*}
\[
\begin{array}{lrcll}
\rlname{Lookup} & (x,\rho)^\ell & \red{0} & v^\ell & \text{where $\rho(x)=v$} \\
\rlname{Apply} & ((\kw{fun}(x)\{t^{\ell_1}\},\rho)^{\ell_2}(v))^{\ell_3} & \red{0} & (t,\rho[x\mapsto v])^{\ell_1} \\
\rlname{Unbox} & (\kw{unbox}\>(\kw{box}\>v^1)^{\ell_1})^{\ell_2} & \red{1} & (v^1)^{\ell_2} \\
\rlname{Run} & (\kw{run}\>(\kw{box}\>v^1)^{\ell_1}\>\kw{in}\>\rho)^{\ell_2} & \red{0} & (v^1, \rho)^{\ell_2} \\
\rlname{IfTrue} & (\kw{if}(\kw{true})\{t_1^{\ell_1}\}\kw{else}\{t_2^{\ell_2}\})^{\ell} & \red{0} & t_1^{\ell_1} \\
\rlname{IfFalse} & (\kw{if}(\kw{false})\{t_1^{\ell_1}\}\kw{else}\{t_2^{\ell_2}\})^{\ell} & \red{0} & t_2^{\ell_2} \\
\rlname{Read1} & (\{\overline{s:v},s_i:v_i,\overline{s:v}'\}^{\ell_1}[s_i^{\ell_2}])^{\ell_3} & \red{0} & v_i^{\ell_3} \\
\rlname{Read2} & (\{\overline{s:v},\mathtt{\quot\_\_proto\_\_\quot}:\{\overline{s:v}'\}^{\ell_1'},\overline{s:v}''\}^{\ell_1}[s_x^{\ell_2}])^{\ell_3} & \red{0} & (\{\overline{s:v}'\}^{\ell_1'}[s_x^{\ell_2}])^{\ell_3} & \text{if $s_x\not\in\overline{s}\cup\overline{s}''$} \\
\rlname{Read3} &  (\{\overline{s:v},\mathtt{\quot\_\_proto\_\_\quot}:\kw{null}^{\ell_1'},\overline{s:v}''\}^{\ell_1}[s_x^{\ell_2}])^{\ell_3} & \red{0} & \kw{undef}^{\ell_3} & \text{if $s_x\not\in\overline{s}\cup\overline{s}''$} \\
\rlname{Write1} & (\{\overline{s:v},s_i:v_i,\overline{s:v}'\}^{\ell_1}[s_i^{\ell_2}]=v_i')^{\ell_3} & \red{0} & \{\overline{s:v},s_i:v_i',\overline{s:v}'\}^{\ell_3} \\
\rlname{Write2} & (\{\overline{s:v}\}^{\ell_1}[s_x^{\ell_2}]=v_x)^{\ell_3} & \red{0} & \{\overline{s:v},s_x:v_x\}^{\ell_3} & \text{if $s_x\not\in\overline{s}$}\\
\rlname{Del1} & (\kw{del}\>\{\overline{s:v},s_i:v_i,\overline{s:v}'\}^{\ell_1}[s_i^{\ell_2}])^{\ell_3} & \red{0} & \{\overline{s:v},\overline{s:v}'\}^{\ell_3} \\
\rlname{Del2} & (\kw{del}\>\{\overline{s:v}\}^{\ell_1}[s_x^{\ell_2}])^{\ell_3} & \red{0} & \{\overline{s:v}\}^{\ell_3} & \text{if $s_x\not\in\overline{s}$}\\
\end{array}
\]
\caption{Labelled proper reduction rules}\label{fig:labelled_proper_reduction_rules}
\end{figure*}

\begin{figure*}
\[
\begin{array}{lrcll}
\rlname{Lift-App} & (((\mm : t^{\ell_1}), \rho)^{\ell_2}(v))^{\ell_3} & \red{0} &
	(\mm : (( t, \rho)^{\ell_1}(v))^{\ell_3})^{\ell_3} \\
\rlname{Lift-If} & (\kw{if} ((\mm : v)^{\ell_0}) \{ t_1^{\ell_1} \} \kw{else} \{ t_2^{\ell_2} \} )^\ell & \red{0} &
	(\mm : (\kw{if} (v) \{ t_1^{\ell_1} \} \kw{else} \{ t_2^{\ell_2} \} )^{\ell})^{\ell} \\
\rlname{Lift-Unbox} & (\kw{unbox}\>(\mm : v)^{\ell_1})^{\ell_2} & \red{1} & (\mm : (\kw{unbox}\>v)^{\ell_2})^{\ell_2} \\
\rlname{Lift-RunIn} & (\kw{run}\>(\mm : v)^{\ell_1} \kw{in}\>\rho)^{\ell_2} & \red{0} &
	(\mm : (\kw{run}\>v \kw{in}\>\rho)^{\ell_2})^{\ell_2} \\
\rlname{Lift-ReadSel} & (v_1 [(\mm : v_2)^{\ell_1}])^{\ell_2} & \red{0} & (\mm : (v_1[v_2])^{\ell_2})^{\ell_2} \\
\rlname{Lift-ReadRec} & ((\mm : v_1)^{\ell_1} [v_2])^{\ell_2} & \red{0} & (\mm : (v_1[v_2])^{\ell_2})^{\ell_2} \\
\rlname{Lift-WriteSel} & (v_1 [(\mm : v_2)^{\ell_1}] = v_3)^{\ell_2} & \red{0} & (\mm : (v_1[v_2] = v_3)^{\ell_2})^{\ell_2} \\
\rlname{Lift-WriteRec} & ((\mm : v_1)^{\ell_1} [v_2] = v_3)^{\ell_2} & \red{0} & (\mm : (v_1[v_2] = v_3)^{\ell_2})^{\ell_2} \\
\rlname{Lift-DelSel} & (\kw{del}\>v_1[(\mm : v_2)^{\ell_1}])^{\ell_2} & \red{0} & (\mm : (\kw{del}\>v_1[v_2])^{\ell_2})^{\ell_2} \\
\rlname{Lift-DelRec} & (\kw{del}\>(\mm : v_1)^{\ell_1}[v_2])^{\ell_2} & \red{0} & (\mm : (\kw{del}\>v_1[v_2])^{\ell_2})^{\ell_2}
\end{array}
\]
\caption{Labelled lifts}\label{fig:labelled_lifts}
\end{figure*}


\subsection{Analysis}

\begin{figure*}
\[
\begin{array}{lrcl}
\text{Abstract values} & \nu \in \AbsVal & ::= & \mathtt{NULL} \mid \mathtt{UNDEF} \mid \mathtt{BOOL} \mid \mathtt{NUM} \mid \mathtt{STR} \\
                       &                 & \mid & \mathtt{FUN}(x,e) \mid \mathtt{BOX}(e) \mid \mathtt{REC}(\ell)\\
\text{Abstract variables} & \xi \in \AbsVar & ::= & x \mid \ell.p \\
\text{Abstract caches} & \Gamma & \colon & \Label\to\pow(\AbsVal) \\
\text{Abstract environments} & \varrho & \colon & \AbsVar\to\pow(\AbsVal) \\
\end{array}
\]
\caption{Abstract domains}\label{fig:abstract-domains}
\end{figure*}

The abstract domains of the analysis are defined in Figure~\ref{fig:abstract-domains}. Abstract variables of the form $x$ represent function parameters; abstract variables of the form $\ell.p$ represent record fields. Note that $e$, $\ell$, $x$ and $p$ only range over expressions, labels and names occurring in the program to be analysed, hence the abstract domains are finite.

For a literal $k$, let $\lceil k\rceil$ be its abstract value, that is:
\[
\begin{array}{lcll}
\lceil\kw{null}\rceil & = & \mathtt{NULL} \\
\lceil\kw{undef}\rceil & = & \mathtt{UNDEF} \\
\lceil b\rceil & = & \mathtt{BOOL} & \text{for boolean $b$} \\
\lceil n\rceil & = & \mathtt{NUM} & \text{for number $n$} \\
\lceil s\rceil & = & \mathtt{STR} & \text{for string $s$} \\
\end{array}
\]

For an abstract environment $\varrho$ and a label $\ell$ we define $\mathtt{proto}(\ell)_{\varrho}$ to be the smallest set $P\subseteq\Label$ such that $\ell\in P$ and for any $p\in P$ and $\mathtt{REC}(\ell')\in\varrho(p.\proto)$ also $\ell'\in P$.

We define three acceptability judgements $\Gamma,\varrho\models e$; $\Gamma,\varrho\models\rho$ and $\Gamma,\varrho\models\nu\approx t$ by mutual induction as shown in
Figure~\ref{fig:acceptability judgements}.

\begin{figure*}
\[
\begin{array}{lll}
\Gamma,\varrho\models k^\ell & \text{if} & \lceil k\rceil\in\Gamma(\ell) \\
\Gamma,\varrho\models x^\ell & \text{if} & \varrho(x)\subseteq\Gamma(\ell) \\
\Gamma,\varrho\models\{\overline{s:e}\}^\ell & \text{if} & \forall i.\Gamma,\varrho\models e_i \\
                                            & \text{and} & \exists\mathtt{REC}(\ell')\in\Gamma(\ell).\forall i.\Gamma(\lbl(e_i))\subseteq\varrho(\ell'.s_i) \\
\Gamma,\varrho\models(\kw{fun}(x)\{e\})^\ell & \text{if} & \Gamma,\varrho\models e \\
                                            & \text{and} & \exists \nu\in\Gamma(\ell).\Gamma,\varrho\models\nu\approx\kw{fun}(x)\{e\} \\
\Gamma,\varrho\models(t_1^{\ell_1}(t_2^{\ell_2}))^{\ell} & \text{if} & \Gamma,\varrho\models t_1^{\ell_1} \land \Gamma,\varrho\models t_2^{\ell_2} \\
                                                    & \text{and} & \forall\mathtt{FUN}(x,t_3^{\ell_3})\in\Gamma(\ell_1).\Gamma(\ell_2)\subseteq\varrho(x)\land\Gamma(\ell_3)\subseteq\Gamma(\ell) \\
\Gamma,\varrho\models(\kw{box}\>e)^\ell & \text{if} & \Gamma,\varrho\models e \\
                                       & \text{and} & \exists \nu\in\Gamma(\ell).\Gamma,\varrho\models\nu\approx\kw{box}\>e \\
\Gamma,\varrho\models(\kw{unbox}\>t^{\ell})^{\ell_0} & \text{if} & \Gamma,\varrho\models t^{\ell} \\
                                                  & \text{and} & \forall\mathtt{BOX}(t'^{\ell'})\in\Gamma(\ell).\Gamma(\ell')\subseteq\Gamma(\ell_0) \\
\Gamma,\varrho\models(\kw{run}\>t^{\ell})^{\ell_0} & \text{if} & \Gamma,\varrho\models t^{\ell} \\
                                                & \text{and} & \forall\mathtt{BOX}(t'^{\ell'})\in\Gamma(\ell).\Gamma(\ell')\subseteq\Gamma(\ell_0) \\
\Gamma,\varrho\models(\kw{run}\>t^{\ell}\>\kw{in}\>\rho)^{\ell_0} & \text{if} & \Gamma,\varrho\models t^{\ell} \land \Gamma,\varrho\models\rho \\
                                                & \text{and} & \forall\mathtt{BOX}(t'^{\ell'})\in\Gamma(\ell).\Gamma(\ell')\subseteq\Gamma(\ell_0) \\
\Gamma,\varrho\models(\kw{if}(t_1^{\ell_1})\{t_2^{\ell_2}\}\kw{else}\{t_3^{\ell_3}\})^{\ell_4} & \text{if} & \Gamma,\varrho\models t_1^{\ell_1} \land \Gamma,\varrho\models t_2^{\ell_2} \land \Gamma,\varrho\models t_3^{\ell_3} \\
                                                                                        & \text{and} & \Gamma(\ell_2)\subseteq\Gamma(\ell_4) \land \Gamma(\ell_3)\subseteq\Gamma(\ell_4) \\
\Gamma,\varrho\models(t,\rho)^\ell & \text{if} & \Gamma,\varrho\models t^\ell \land \Gamma,\varrho\models\rho \\
\Gamma,\varrho\models(t_1^{\ell_1}[t_2^{\ell_2}])^\ell & \text{if} & \Gamma,\varrho\models t_1^{\ell_1} \land \Gamma,\varrho\models t_2^{\ell_2} \\
                                              & \text{and} & \forall\mathtt{REC}(\ell')\in\Gamma(\ell_1).\forall s,\ell''\in\mathtt{proto}(\ell')_{\varrho}.\rho(\ell''.s)\subseteq\Gamma(\ell) \\
                                              & \text{and} & \mathtt{UNDEF}\in\Gamma(\ell) \\
\Gamma,\varrho\models(t_1^{\ell_1}[t_2^{\ell_2}]=t_3^{\ell_3})^\ell & \text{if} & \Gamma,\varrho\models t_1^{\ell_1} \land \Gamma,\varrho\models t_2^{\ell_2} \land \Gamma,\varrho\models t_3^{\ell_3} \\
                                                              & \text{and} & \forall s,\mathtt{REC}(\ell')\in\Gamma(\ell_1).\Gamma(\ell_3)\subseteq\varrho(\ell'.s) \\
                                                              & \text{and} & \Gamma(\ell_1)\subseteq\Gamma(\ell) \\
\Gamma,\varrho\models(\kw{del}\>t_1^{\ell_1}[t_2^{\ell_2}])^\ell & \text{if} & \Gamma,\varrho\models t_1^{\ell_1} \land \Gamma,\varrho\models t_2^{\ell_2} \\
                                                              & \text{and} & \Gamma(\ell_1)\subseteq\Gamma(\ell) \\
\Gamma,\varrho\models(\mm : t_1^{\ell_1})^\ell & \text{if} & \Gamma,\varrho\models t_1^{\ell_1} \\
                                                              & \text{and} & \Gamma(\ell_1)\subseteq\Gamma(\ell) \\
\\
\Gamma,\varrho\models\rho & \text{if} & \forall x\in\dom(\rho).\Gamma,\varrho\models\rho(x)\land\Gamma(\lbl(\rho(x)))\subseteq\varrho(x) \\
\\
\Gamma,\varrho\models\lceil k\rceil\approx k & \text{for any literal $k$} & \\
\Gamma,\varrho\models\mathtt{FUN}(x,e)\approx \kw{fun}(x)\{e\} \\
\Gamma,\varrho\models\mathtt{BOX}(e)\approx \kw{box}\>e \\
\Gamma,\varrho\models\mathtt{REC}(\ell')\approx \{\overline{s:t^\ell}\} & \text{if} & \forall i.\exists \nu_i\in\varrho(\ell'.s_i). \Gamma,\varrho\models\nu_i\approx t_i\\
\Gamma,\varrho\models\nu\approx t' & \text{if} & \Gamma,\varrho\models \nu\approx t \land t^\ell\ev{n}t'^{\ell} \\
\Gamma,\varrho\models\nu\approx(t,\rho) & \text{if} & \Gamma,\varrho\models\nu\approx t \land \Gamma,\varrho\models\rho \\
\end{array}
\]
\caption{Acceptability judgements}\label{fig:acceptability judgements}
\end{figure*}

\end{document}